\documentclass[12pt,a4paper,reqno,oneside]{amsart}
\usepackage{t1enc}
\usepackage{a4wide}
\usepackage{times}
\usepackage{amssymb}
\usepackage[mathscr]{euscript}
\usepackage{graphicx}

\newtheorem{thm}{Theorem}[section]
\newtheorem{prop}[thm]{Proposition}
\newtheorem{lem}[thm]{Lemma}
\newtheorem{cor}[thm]{Corollary}

\theoremstyle{remark}
\newtheorem{rem}[thm]{Remark}

\newtheorem{ex}[thm]{Example}

\theoremstyle{definition}
\newtheorem{defn}[thm]{Definition}

\renewcommand{\phi}{\varphi} 
\newcommand{\tr}{\mathrm{Tr}} 
\newcommand{\op}{\mathrm{op}} 

\newcommand{\C}{\mathcal C} 
\newcommand{\E}{\mathrm{E}} 
\newcommand{\T}{\mathcal T} 
\newcommand{\U}{\mathcal U} 

\renewcommand{\L}{\mathcal L}

\newcommand{\<}{\langle}
\renewcommand{\>}{\rangle}
\renewcommand{\Re}{\mathrm{Re}} 
\newcommand{\Q}{\mathcal Q} 

\newcommand{\ran}{\mathrm{ran}} 
\newcommand{\dom}{\mathrm{dom}} 
\newcommand{\fraum}{(\Omega,(\mathcal F_t)_{t\in\mathbb R_+},\mathfrak A,P)} 
\newcommand{\vektor}[2]{\left(\begin{matrix} #1\\ #2\end{matrix}\right)}

\newcommand{\nooutput}[1]{}

\begin{document}
\date{Version of \today}
\title{Representation of infinite dimensional forward price models in commodity markets}
\author[Benth]{Fred Espen Benth}
\address[Fred Espen Benth]{\\
Centre of Mathematics for Applications \\
University of Oslo\\
P.O. Box 1053, Blindern\\
N--0316 Oslo, Norway}
\email[]{fredb\@@math.uio.no}
\urladdr{http://folk.uio.no/fredb/}
\author[Kr\"uhner]{Paul Kr\"uhner}
\address[Paul Kr\"uhner]{\\
Department of Mathematics \\
University of Oslo\\
P.O. Box 1053, Blindern\\
N--0316 Oslo, Norway}
\email[]{paulkru\@@math.uio.no}

\keywords{Forward price, Infinite dimensional stochastic processes, L\'evy processes, commodity markets, Heath-Jarrow-Morton approach}

\thanks{This paper has been developed under financial support of the project "Managing Weather Risk in
Electricity Markets" (MAWREM), funded by the RENERGI-program of the Norwegian Research Council}

\nooutput{
We study the forward price dynamics in commodity markets realized as a process
with values in a family of Hilbert spaces of absolutely continuous functions defined by Filipovi\'c~\cite{filipovic.01}. The forward dynamics is defined as the mild solution of a
certain stochastic partial differential equation driven by an infinite dimensional 
L\'evy process. It is shown that the associated spot price dynamics can 
be expressed as a sum of Ornstein-Uhlenbeck processes, or more generally,  as
a sum of certain stationary processes. These results link the possibly infinite dimensional
forward dynamics to classical commodity spot models. We continue with a detailed analysis
of multiplication and integral operators on the Hilbert spaces and show that 
Hilbert-Schmidt operators are essentially integral operators. The covariance operator of the 
L\'evy process driving the forward dynamics and the diffusion term can both be specified in
terms of such operators, and we analyse in several examples the consequences on model
dynamics and their probabilistic properties. Also, we represent the forward price for contracts delivering over a period in terms of an integral operator, a case being relevant for power and gas markets.
In several examples we reduce our general model to existing commodity spot and 
forward dynamics.  
}

\begin{abstract}
 We study the forward price dynamics in commodity markets realized as a process with values in a Hilbert space of absolutely continuous functions defined by Filipovi\'c~\cite{filipovic.01}. 
The forward dynamics are defined as the mild solution of a certain stochastic partial differential equation driven by an infinite dimensional L\'evy process. 
It is shown that the associated spot price dynamics can be expressed as a sum of Ornstein-Uhlenbeck processes, or more generally,  as a sum of certain stationary processes. 
These results link the possibly infinite dimensional forward dynamics to classical commodity spot models. We continue with a detailed analysis of multiplication and integral operators on the Hilbert spaces and show that Hilbert-Schmidt operators are essentially integral operators. 
The covariance operator of the L\'evy process driving the forward dynamics and the diffusion term can both be specified in terms of such operators, and we analyse in several examples the consequences on model dynamics and their probabilistic properties. 
Also, we represent the forward price for contracts delivering over a period in terms of an integral operator, a case being relevant for power and gas markets.
In several examples we reduce our general model to existing commodity spot and forward dynamics. 
\end{abstract}

\maketitle

\section{Introduction}
In this paper we analyse the dynamics of forward prices in commodity markets in 
the context of infinite dimensional stochastic calculus. The motivation 
for our studies comes mainly from energy markets like power and gas where
one finds strong seasonal patterns, a high degree of idiosyncratic risk over different market segments as well as non-Gaussian features like spikes.

In mathematical finance the arbitrage-free forward price can be derived from the 
buy-and-hold strategy in the underlying spot commodity (see Hull~\cite{H}).  
The forward price dynamics is thus implied from a given stochastic model of the spot
commodity. If one specifies the price of the spot commodity by a semimartingale 
dynamics, one can represent the forward price as the conditional 
expected value of the spot at time of delivery, where the expectation is calculated with
respect to an equivalent martingale measure $Q$.
However, electricity is non-storable and hence not tradeable in the classical sense. 
Consequently, the buy-and-hold hedging strategy breaks down, and the spot price process does not need to be a $Q$-martingale itself. Hence, any equivalent measure $Q$ can be used as a 
pricing measure when deriving an {\it arbitrage-free} forward price. This raises the question of finding a financially valid $Q$.  

As the fundamental relationship between the spot and forward is highly delicate in 
energy markets, it seems like a reasonable alternative to model the forward price dynamics directly. The alternative approach of modelling derivatives directly was first advocated for interest rate markets in Heath, Jarrow and Morton~\cite{heath.al.92}, commonly referred to as the HJM approach, and the idea of this novel method has been transferred to other markets, see e.g. 
Carmona and Nadtochiy~\cite{carmona.nadtochiy.12}, Kallsen and Kr\"uhner~\cite{kallsen.kruehner.13}, and B\"uhler~\cite{buehler.06}.

Due to its mathematical nature it is straightforward to transfer the HJM approach to 
commodity forward markets. This was first done by Clewlow and Strickland~\cite{CS},
and later analysed in the context of power markets by Benth and Koekebakker~\cite{BK}. 
However, both these works used forward curve models driven by finite dimensional 
Wiener noise (in fact, one-dimensional!). There are ample empirical evidence for a 
high degree of idiosyncratic risk in power forward markets, in particular. 
Andresen et al.~\cite{andresen.et.al.10} studied the Nordic power market NordPool, and found a clear correlation 
structure between contracts with different time to delivery. Earlier studies of
the same market by Koekebakker and Ollmar~\cite{KO} using principal component analysis also points towards a very high dimension of the noise. Moreover, marginal price behaviour is strongly non-Gaussian, as the study of Frestad, Benth and
Koekebakker~\cite{benth.et.al.10} clearly indicates.  These insights from real market prices 
motivates an HJM approach based on infinite dimensional non-Gaussian noise. 

The function spaces introduced by Filipovi\'c~\cite{filipovic.01} provide a natural framework to 
model the forward curve evolution. These spaces have been applied in fixed-income markets for the forward {\it rate} dynamics, cf.\ \cite{carmona.tehranchi.06}. Denoting by
$F(t,T)$ the forward price at time $t\leq T$ for a contract delivering a spot commodity at 
time $T$, the HJM approach gives the dynamics of $F$ as the solution of a stochastic differential equation
$$ 
dF(t,T) = \beta(t,T)dt + \Psi(t,T)dL(t),
$$
for  some infinite dimensional L\'evy process $L$ and appropriately defined parameters
$\beta$ and $\Psi$. Since the forward curve $F(t,T)_{T\geq t}$ changes its domain over
 time it is  more convenient to work with the Musiela parametrisation
$$ 
f(t,x) := F(t,t+x),\quad x\geq 0. 
$$
We interpret $x=T-t$ as {\it time to delivery}, whereas $T$ is {\it time of delivery}. Then,
heuristically speaking, the forward curve follows the stochastic partial differential equation (SPDE)
\begin{equation}
\label{eq:spde} 
df(t,x) = (\beta(t,t+x)+\partial_xf(t,x))dt + \Psi(t,t+x)dL(t).
\end{equation}
However, in order to make sense of this SPDE it is useful to work in an appropriate space of functions which contains the entire forward curve $(f(t,x))_{x\geq 0}$ for any time $t\geq0$. 
For this purpose we shall use the spaces introduced by Filipovi\'c~\cite{filipovic.01}.

We remark that the spot price of the commodity is recovered by passing to the limit
for $x\downarrow 0 $ in $f(t,x)$. Moreover, if we model the forward curve directly
under the pricing measure $Q$ which is the common strategy in the HJM approach, 
martingale conditions must be imposed in the dynamics. It is to be noted that the forward
contracts are liquidly traded financial assets in most commodity markets, and therefore it
may sometimes be natural to model the dynamics under the objective (market) probability $P$. 

In this paper we provide a detailed analysis of the forward curve dynamics. In particular,
we are able to derive various representations of the joint forward price dynamics for 
a given finite set of contracts, and relate this to finite-dimensional models. More
specifically, we recover the covariance structure from the covariance operator
of the infinite dimensional noise of the initially given forward curve dynamics. 
Furthermore, we prove that the spot price dynamics in special cases can be represented as
a sum of Ornstein-Uhlenbeck processes. Such mean-reverting processes are popular in
commodity markets to model the stationary evolution of prices, reflecting 
the direct influence of demand and supply (see Benth, \v{S}altyt\.{e} Benth and 
Koekebakker~\cite{BSBK-book} for a discussion and references). More generally, we obtain
a spot price dynamics in terms of so-called L\'evy semistationary processes, which have 
gained recent attention in power markets (see Garcia, Kl\"uppelberg and M\"uller ~\cite{GKM}
for the special case of CARMA processes, and Barndorff-Nielsen, Benth 
and Veraart~\cite{BNBV-spot}). These processes are also natural in modelling the temperature
dynamics, and thus are relevant to weather forward prices.   

We also make a detailed study of operators on the spaces introduced by Filipovi\'c~\cite{filipovic.01}. We have a particular attention on multiplication and integral 
operators, as these are natural objects when specifying $\Psi$ in 
\eqref{eq:spde}, and when defining the covariance operator associated with 
(square-integrable) infinite dimensional L\'evy processes. We give complete characterizations
of these operators, where we specifically show that all Hilbert-Schmidt operators
essentially are integral operators. In examples we link these operators to concrete 
models relevant for the forward dynamics in energy markets.  

In power and gas, and in fact also weather derivative markets, the forward contracts
are typically settled over a {\it delivery period} and not at a specific delivery time. For example, in the NordPool market, forward contracts are settled financially on the hourly spot 
price over different delivery periods that range from a day to a year. As an application of
our study of operators, we identify forward contracts with a delivery period as
an operator on elements in the space defined by Filipovi\'c~\cite{filipovic.01} mapping
onto itself. Hence, the analytic properties of a fixed-delivery forward curve can be transported
to a delivery-period forward curve by an operator, being in fact a sum of the identity operator and an explicitly given integral operator.

Our results are presented as follows: In the next section we find various representations for stochastic processes which are linearly mapped to finite dimensional spaces. Section three contains our main discussion on forward curve models for commodity markets where we derive various representations for the resulting spot price process, the correlation of forward price process and for the forward curve itself. In the last section we analyse integral operators, Hilbert Schmidt operators and multiplication operators on the space defined by Filipovi\'c~\cite{filipovic.01} in detail.

%
%
%
%
%

\subsection{Notation}
$\mathbb R$, resp.\ $\mathbb C$, denotes the real, resp.\ the complex numbers, and $\mathbb R_+:=[0,\infty)$ (resp.\ $\mathbb R_-:=(-\infty,0]$) the non-negative (resp.\ non-positive) real numbers. $\fraum$ will always denote a right-continuous filtered probability space. We denote the set of $p$-integrable functions from a measure space $(\Sigma,\mathcal C,\mu)$ to $\mathbb R$ by $L^p((\Sigma,\mu))$ and if there is no ambiguity with the measure, then we simply write $L^p(\Sigma)$ instead. For an operator $\mathcal{T}\in L(H_1,H_2)$,
$L(H_1,H_2)$ being the space of linear operators between the two Hilbert spaces $H_1$ and $H_2$, we denote the dual operator of $\mathcal{T}$ by $\mathcal{T}^*$. The set of positive opertors with finite trace on a Hilbert space $H$ is denote by $L^+_1(H)$. For separable Hilbert spaces $U$, $H$ and an $U$-valued square integrable  L\'evy process $L$ with martingale covariance $\Q$ relative to some filtered probability space we denote by $\mathcal{L}_{L,T}^2(H)$ the set of predictable $L(U,H)$-valued processes such that
$$ \int_0^T \tr\left( \Psi(s)\Q\Psi(s)^*) \right)ds<\infty,\quad t\geq 0, $$
and $\mathcal L_L^2(H) := \bigcap\{\mathcal L_{L,T}^2(H):T>0\}$, cf.\ Peszat and Zabczyk~\cite[formula (8.6)]{peszat.zabczyk.07} where $\tr$ denotes the trace of the operator.
Further unintroduced notations are used as in Peszat and Zabczyk~\cite{peszat.zabczyk.07}.

\section{Representation of functionals of stochastic integrals in a Hilbert space}
\label{section:gen-repr-integrals}

In this Section we derive some general results on the representation of linear functionals
of stochastic integrals in a Hilbert space. These results will be useful in
our analysis of the forward curve dynamics in Section~\ref{forward-sect}, but are interesting
in their own right as well.

Let us consider a stochastic integral $Y$ of the form
$$
Y(t) = \int_0^t\Psi(s)dW(s)\,,
$$
for some Brownian motion $W$ with values in a separable Hilbert space $U$ and an
integrable stochastic process $\Psi:\mathbb R_+\times U\mapsto H$, with $H$ a 
separable Hilbert space
(see Peszat and Zabczyk~\cite{peszat.zabczyk.07} for conditions). 
Sometimes one is only interested in one-dimensional marginals, i.e.\ one is interested in
$$
X(t) := \mathcal{T}(Y(t))\,,
$$
where $\mathcal{T}$ is a continuous linear functional on the state space $H$ of $Y$. 
Then, of course, we have
$$
X(t) =\int_0^t(\mathcal{T}\circ\Psi(s))dW(s)\,.
$$
If $U$ is finite dimensional, it is well-known that there is some standard (real-valued)
Brownian motion $B$ and some It\^o integrable stochastic process $\sigma$ such that
$$
X(t) = \int_0^t\sigma(s)dB(s)\,.
$$
The next Theorem shows that a similar representation holds if $U$ is any separable Hilbert space:

\begin{thm}\label{S:endl. Darstellung, stetig}
Let $n\in\mathbb N$ and $H,U$ be separable Hilbert spaces. Let $W$ be a square integrable and mean zero $U$-valued Wiener process with covariance $\Q\in L^+_1$. Assume that $\dim\ran(\Q) \geq n$ and $\Q$ is positive definite. Let $\Psi\in\mathcal L^2_{W}(H)$, 
$\mathcal{T}\in L(H,\mathbb R^n )$ and define
$$
X(t):= \mathcal{T}\left(\int_0^{t}\Psi(s)dW(s)\right)\,.
$$
Then there is an $n$-dimensional standard Brownian motion $B$ such that
$$
X(t) = \int_0^t\sigma(s)dB(s)\,,
$$
where $\sigma(s):=(\mathcal{T}\Psi(s)\Q\Psi(s)^*\mathcal{T}^*)^{1/2}\in\mathcal L^2_B(\mathbb R^{n})$. If $\sigma(s)$ is invertible in $\mathbb R^{n\times n}$ for $\lambda\otimes P$-almost any $s\in\mathbb R$, then $\sigma^{-1}\in\mathcal L^2_{X}(\mathbb R^n)$ and
$$ B(t) = \int_0^t(\sigma(s))^{-1}dX(s).$$
\end{thm}
\begin{proof}
 Peszat and Zabczyk~\cite[Theorem 8.7(v)]{peszat.zabczyk.07} yield
$$
X(t) = \int_0^{t}\Gamma(s)dW(s)\,,
$$
where $\Gamma(s) := \T\circ\Psi(s)$. 
Peszat and Zabczyk~\cite[Corollary 8.17]{peszat.zabczyk.07} imply
$$
\<\<X,X\>\>_t = \int_0^t \Gamma(s)\Q\Gamma(s)^*ds\,.
$$  
Corollary \ref{C:Pfadeigenschaft} yields that $X$ has a.s.\ continuous paths. The assumptions assure that there is an $n$-dimensional standard Brownian motion $B$ on the filtered probability space $\fraum$. Thus Jacod~\cite[Corollaire 14.47(b)]{jacod.79} yields that there is an $n$-dimensional standard Brownian motion $B$ such that
\begin{eqnarray}\label{e:alternativedarstellung}
X(t) = \int_0^t\sigma(s)dB(s)\,.
\end{eqnarray}

Now assume that $\sigma(s)$ is an invertible element of $\mathbb R^{n\times n}$ for $\lambda\otimes P$-almost any $s\in\mathbb R$, denote the matrix inverse of $\sigma(s)$ by $\gamma(s)$ and denote the identity matrix on $\mathbb R^n$ by $I_n$. The representation \eqref{e:alternativedarstellung} and Peszat and Zabczyk~\cite[Definition 8.3, Theorem 8.7(iv)]{peszat.zabczyk.07} yield that the martingale covariance $\Q^X$ of $X$ is given by 
$$\Q^X_t = 1_{\{\tr(\sigma(t)\sigma(t)^*)\neq 0\}}\frac{\sigma(t)\sigma(t)^*}{\tr(\sigma(t)\sigma(t)^*)},\quad t\geq0.$$
Then we have
$$ \E \left(\int_0^T \tr\left(\gamma(s)\Q_s^X\gamma(s)^*\right) d\<X,X\>(s) \right)  = nT <\infty$$
and hence $\gamma \in \mathcal L^2_X(\mathbb R^n)$ by definition, cf.
Peszat and Zabczyk~\cite[page 113]{peszat.zabczyk.07}.

Hence 
\begin{eqnarray*}
 B(t) &=& \int_0^t \sigma^{-1}(s)\sigma(s)dB(s) \\ 
 &=& \int_0^t\sigma^{-1}(s) dX(s).
\end{eqnarray*}
\end{proof}

Remark that if $n=1$, then there is a stochastic process $\gamma:=(\T\Psi)^*1$ such that $(\T\Psi)x = \<\gamma,x\>$ for any $x\in U$. If the set  
$$
\{(\omega,t)\in\Omega\times\mathbb R_+:\Q\gamma(\omega,t)=0\}
$$
is a $P\otimes\lambda$-null set, then $\varphi:=\frac{\<\gamma,\cdot\>}{\<\Q\gamma,\gamma\>^{1/2}}\in\mathcal L^2_W(U)$ and the standard Brownian motion given in the theorem above is also given by
$$
B(t)= \int_0^t\varphi(s)dW(s)\,.
$$
Moreover, then we have
$$
\T\left(\int_0^t\Psi(s)dW(s)\right)=\int_0^t\<\gamma(s),dW(s)\> = \int_0^t\<\Q\gamma(s),\gamma(s)\>^{1/2}dB(s)\,.
$$

We extend the result in Theorem~\ref{S:endl. Darstellung, stetig} to the class 
of L\'evy processes defined as subordinated Brownian motions in the next Subsection.

\subsection{The case of subordinated Brownian motion}
Let us recall from Benth and Kr\"uhner~\cite[Theorem 4.7]{benth.kruehner.13} the definition of a subordinated 
Brownian motion:
\begin{defn}
 A {\em subordinated Brownian motion} $L$ with values in some Hilbert space $U$ is a L\'evy process such that there is a $U$-valued Brownian motion $W$ and a subordinator $\Theta$ which is independent of $W$ such that
 $$L(t) = W(\Theta(t)).$$
 Subordinated Brownian motions $L$, $N$ are of the {\em same type} if there are Brownian motions $W_L$, $W_N$ and subordinators $\Theta_L$, $\Theta_N$ such that $W_N$, $\Theta_N$ are independent, $W_L,\Theta_L$ are independent, $\Theta_L$, $\Theta_N$ have the same law and
 \begin{eqnarray*}
   L(t) &=& W_L(\Theta_L(t)) \\
   N(t) &=& W_N(\Theta_N(t))
 \end{eqnarray*}
 for any $t\in\mathbb R_+$.
\end{defn}

Subordinated Brownian motions have some similarities with Brownian motion if the subordinator has finite first moment. In particular, the set of integrands can be compared 
easily. First, we recall the notion of time changed filtrations:
\begin{defn}
 A {\em time change} is a right-continuous increasing family $(\theta(t))_{t\in\mathbb R_+}$ of stopping times with respect to some right-continuous filtration $(\mathcal F_t)_{t\in\mathbb R_+}$. The {\em time changed filtration} is the filtration given by
 $$\mathcal F^\theta_t := \mathcal F_{\theta(t)} = \bigcap_{s>t} \mathcal F_{\theta(s)}\quad t\in\mathbb R_+.$$
 Let $X$ be an $\mathcal F$-adapted stochastic process and assume that the time-change is finite valued. The {\em time-changed} process is given by $X^\theta(t):=X(\theta(t))$, $t\in\mathbb R_+$.
\end{defn}
We have the following result on stochastic integration with respect to subordinated 
Brownian motions.
\begin{lem}\label{l:Zeittransformiertes integral}
 Let $W$ be a mean-zero Brownian motion with values in a separable Hilbert space $U$ relative to some filtration $(\mathcal F_t)_{t\in\mathbb R_+}$, $H$ be another separable Hilbert space, $\psi\in\mathcal L^2_L(\mathcal G,H)$ a measurable, $\Theta$ be a non-zero subordinator with finite moment such that $\Theta(t)$ is a stopping time for each $t\in\mathbb R_+$. Let $(\mathcal G_t)_{t\in\mathbb R_+}$ be the time-changed filtration given by $\mathcal G_t:=\mathcal F_{\Theta(t)}$ and $L(t):=W(\Theta(t))$, $t\in\mathbb R_+$. Then $L$ is a $U$-valued square integrable L\'evy process and there is an isometric embedding
 $$\Gamma:\mathcal L^2_L(\mathcal G,H)\rightarrow \mathcal L_W^2(\mathcal F,H)$$
 such that
 \begin{equation}\label{e:Gamma psi 1}
 \int_0^t \psi(s) dL(s) = \int_0^{\Theta(t)} \Gamma(\psi)(s) dW(s)\quad P\text{-a.s.},
\end{equation}
 where the left stochastic integral is with respect to the filtration $(\mathcal G_t)_{t\in\mathbb R_+}$, the right stochastic integral is with respect to the filtration $(\mathcal F_t)_{t\in\mathbb R_+}$.
\end{lem}
\begin{proof}
 $\Theta$ is a time-change on the filtration $(\mathcal F_t)_{t\in\mathbb R_+}$ by definition. Benth and Kr\"uhner~\cite[Theorem 4.7]{benth.kruehner.13} yields that $L$ is a L\'evy process in its own right, however its increments are independent of the time-changed filtration and thus it is L\'evy relative to the filtration $(\mathcal G_t)_{t\in\mathbb R_+}$. In order to establish the isometric embedding $\Gamma$ it is sufficient to work with elementary integrands. Let $\psi\in \mathcal L^2_L(\mathcal G,H)$ be elementary, i.e.\ there is $n\in\mathbb N$, $0\leq a_j\leq b_j<\infty$, $\mathcal G_{a_j}$-measurable square integrable random variabes $Y_j$ and $\phi_j\in L(U,H)$ such that
 $$\psi = \sum_{j=1}^n Y_j 1_{]a_j,b_j]}\phi_j.$$
  Define $\psi_\Theta := \sum_{j=1}^n Y_j 1_{]\Theta(a_j),\Theta(b_j)]}\phi_j$. Observe that $\psi_\Theta$ does not depend on the representation of $\psi$. 
 We have $\psi_\Theta\in \mathcal L_W^2(\mathcal F,H)$ and 
 \begin{align*}
   \int_0^t\psi(s) dL(s) &= \sum_{j=1}^n Y_j \phi_j\left(L(t\wedge b_j)-L(t\wedge a_j)\right) \\
                         &= \sum_{j=1}^n Y_j \phi_j\left(W(\Theta(t)\wedge \Theta(b_j))-W(\Theta(t)\wedge \Theta(a_j))\right) \\
                         &= \int_0^{\Theta(t)} \psi_\Theta(s) dW(s).
 \end{align*}
 Moreover, we have
 \begin{align*}
  \int_0^\infty \E(\<\Q^L\psi(s),\psi(s)\>) d\<L,L\>(s) &= \E\left(\left(\int_0^\infty\psi(s)dL(s)\right)^2\right) \\
                  &=  \E\left(\left(\lim_{n\rightarrow\infty}\int_0^{\Theta(n)} \psi_\Theta(s) dW(s)\right)^2\right) \\
                  &=  \E\left(\left(\int_0^{\infty} \psi_\Theta(s) dW(s)\right)^2\right) \\
                  &= \int_0^\infty \E(\<\Q^W\psi_\Theta(s),\psi_\Theta(s)\>) d\<W,W\>(s)
 \end{align*}
 where $\Q^L$ denotes the martingale covariance of $L$ and 
$\Q^W$ denotes the martingale covariance of $W$.

 From the construction we can see that Equation \eqref{e:Gamma psi 1} holds for elementary integrands, and thus for all integrands by a density argument.
%
\end{proof}
By appealing to this Lemma we can derive a connection between functionals of the
infinite dimensional stochastic integral and finite dimensional versions of it. 
\begin{thm}\label{S:endl. Darstellung, subordiniert}
Let $n\in\mathbb N$ and $H,U$ be separable Hilbert spaces. Let $L$ be a non-zero, square integrable mean zero $U$-valued subordinated Brownian motion with martingale covariance $\Q\in L^+_1$. Assume that $\dim\ran(\Q) \geq n$. Let $\Psi\in\mathcal L^2_{L}(H)$, $\T\in L(H,\mathbb R^n )$ and define
$$
X(t) := \T\left(\int_0^{t}\Psi(s)dL(s)\right)\,.
$$
Then there is an $n$-dimensional square integrable mean-zero L\'evy process $N$ such that
$$X(t) = \int_0^t\sigma(s)dN(s)$$
where $\sigma(s):=(\T\Psi(s)\Q\Psi(s)^*\T^*)^{1/2}\in\mathcal L^2_N(\mathbb R^{n})$. If $\sigma(s)$ is invertible in $\mathbb R^{n\times n}$ for $\lambda\otimes P$-almost any $s\in\mathbb R$, then $\sigma^{-1}\in\mathcal L^2_{X}(\mathbb R^n)$ and
$$ N(t) = \int_0^t(\sigma(s))^{-1}dX(s).$$
\end{thm}

\begin{proof}
First we assume that $\Psi$ is elementary, i.e.\ there are $n\in\mathbb N$, $0\leq a_j\leq b_j<\infty$, $\mathcal F_{a_j}$-measurable square integrable random variabes $Y_j$ and $\phi_j\in L(U,H)$ such that
 $$\Psi = \sum_{j=1}^n Y_j 1_{]a_j,b_j]}\phi_j.$$

By definition of $L$ we have $L(t)=W(\Theta(t))$ for a $U$-valued Wiener process
$W$ with martingale covariance $\Q^W$ and a subordinator $\Theta$. Let $\Gamma$ be the isometric embedding given in Lemma \ref{l:Zeittransformiertes integral}. Note that $\Gamma(\T\Psi) = \T\Gamma(\Psi)$ because this holds if $\Psi$ is elementary. 
Peszat and Zabczyk~\cite[Theorem 8.7(v)]{peszat.zabczyk.07}, Lemma \ref{l:Zeittransformiertes integral} and Theorem \ref{S:endl. Darstellung, stetig} yield
 \begin{align*}
  \T\left(\int_0^t \Psi(s) dL(s)\right) &= \int_0^t \T\Psi(s) dL(s) \\
                          &= \int_0^{\Theta(t)} \Gamma(\T\Psi)(s) dW(s) \\
                          &= \int_0^{\Theta(t)} \T\Gamma(\Psi)(s) dW(s) \\
                          &= \int_0^{\Theta(t)} \sigma_1(s) dB(s)
 \end{align*}
where $B$ is a standard Brownian motion on $\mathbb R^n$ and 
\begin{align*}
 \sigma_1(t)  &= \left(\T\Gamma(\Psi)(t)\Q^W\Gamma(\Psi)^*(t)\T^*\right)^{1/2}. \\
            &= \Gamma\left(\left(\T\Psi\Q^W\Psi^*\T^*\right)^{1/2}\right)(t)
\end{align*}
Applying Lemma \ref{l:Zeittransformiertes integral} again yields
$$ \T\left(\int_0^t \Psi(s) dL(s)\right) = \int_0^t\sigma(s)dN(s) $$
where $N(t):=B(\Theta(t))$, $t\in\mathbb R_+$ and
 $$ \sigma(t) = \left(\T\Psi(t)\Q^W\Psi(t)^*\T^*\right)^{1/2}. $$
Sato~\cite[Theorem 30.1]{sato.99} implies that $N$ is an $n$-dimensional
L\'evy process with the desired properties. Observe that the martingale covariance $\Q^L$ of $L$ and $\Q^W$ coincide.

The general case can be deduced by a density argument.
\end{proof}
We consider an example: Let $\Theta$ be an inverse Gaussian subordinator, that is,
a L\'evy process with increasing paths where $\Theta(1)$ is inverse Gaussian
distributed. Then, $L(t)=W(\Theta(t))$ is a
Hilbert-space valued normal inverse Gaussian (HNIG) L\'evy process in the sense of Benth and Kr\"uhner~\cite[Definition 4.1]{benth.kruehner.13}. The L\'evy process $N$ in the theorem above becomes an $n$-variate normal inverse Gaussian (MNIG) L\'evy process (see 
Rydberg~\cite{R}). In particular, in the case $n=1$ it is an ordinary NIG L\'evy process
(see Barndorff-Nielsen~\cite{BN}). NIG L\'evy processes are frequently used for modelling asset prices in various financial markets, including commodities (see B\"orger et al.~\cite{BCKS}) and power (Andresen, Koekebakker and Westgaard \cite{andresen.et.al.10}, and Benth, Frestad, Koekebakker \cite{benth.et.al.10}). HNIG L\'evy processes are relevant in
the context of modelling the dynamics of the power forward curve.


After these general considerations, we shall now focus on the dynamics of 
forward prices and various representations in the remainder of the paper.


\section{Forward curve modelling and representations}
\label{forward-sect}

We are interested in studying the dynamics of the forward price $F(t,T)$, $0\leq t\leq T$, 
of a contract delivering a commodity at time $T>0$. We shall mostly be concerned with the
more convenient representation of $F$ in terms of the Musiela parametrisation, where we 
let
$$
F(t,T)=f(t,T-t)\,,
$$
for a function $f(t,x)$, $x\geq 0$, denoting the forward price at time $t$ for a 
contract with {\it time to delivery} $x$. We shall interpret $f$ as a
stochastic process with values in a space of functions on $\mathbb R_+$.   

More specifically, we consider the Hilbert spaces $H_w$ introduced by 
Filipovi\'c~\cite[chapter 5]{filipovic.01}. These Hilbert-spaces are 
suitable for modelling the {\it forward rate} curves in the Musiela parametrisation, see 
e.g.~Filipovi\'c, Teichmann and Tappe~\cite{filipovic.et.al.09}. As the modelling of forward prices are similar in
idea to the context of forward rates in fixed-income markets, we adopt the spaces
$H_w$ for our analysis. 

\subsection{Some preliminary results on $H_w$}
Let us introduce and recall some basic facts of the spaces $H_w$ before turning our 
attention to forward curves and their dynamics. We remark that the following results on
$H_w$ can be essentially found in Filipovi\'c~\cite[chapters 4,5]{filipovic.01}. However, some of the results
are slight modifications where proofs have been added.

Let $\mathrm{AC}(\mathbb R_+,\mathbb R)$ be the set of functions $g:\mathbb R_+\rightarrow\mathbb R$ which are absolutely continuous. For a function $w:\mathbb R_+\rightarrow [1,\infty)$ being continuous, increasing and $w(0)=1$, define
$$H_{w}:=\left\{g\in \mathrm{AC}(\mathbb R_+,\mathbb R):\int_{0}^\infty w(x)g'(x)^2dx<\infty\right\}.$$
We introduce the inner product $\<g,h\>:=g(0)h(0)+\int_0^\infty w(x)g'(x)h'(x)dx$ and the corresponding norm $\Vert g\Vert_w := \sqrt{\<g,g\>}$ for any $g,h\in H_{w}$. Note that all the following results also hold if $\mathrm{AC}(\mathbb R_+,\mathbb R)$ is replaced by $\mathrm{AC}(\mathbb R_+,\mathbb C)$ which however has no financial meaning.




We observe that the point evaluation $\delta_x:H_w\rightarrow \mathbb R$ for any 
$x\in\mathbb R$, where $\delta_x(g)=g(x)$, is a continuous linear functional on $H_w$. Its dual operator can be explicitly characterized, as is shown in the following Lemma:
\begin{lem}\label{l:stetigkeit der Punktauswertung}
For $x\in\mathbb R$, define 
$$h_x:\mathbb R_+\rightarrow\mathbb R,y\mapsto 1+\int_0^{y\wedge x}\frac{1}{w(s)}ds$$
Then the dual operator of $\delta_x$ is given by
$$\delta_x^*:\mathbb R\rightarrow H_w,c\mapsto ch_x$$
and its operator norm is given by
$$\Vert\delta_x\Vert^2_{\op} = h_x(x).$$
\end{lem}
\begin{proof}
The first statement is a special case of Filipovi\'c~\cite[Lemma 5.3.1]{filipovic.01}.
Next,
since the operator norm coincides with the norm of $h_x$ we get
\begin{eqnarray*}
 \Vert \delta_x\Vert^2_{\op} &=& h_x(0)h_x(0) + \int_0^\infty w(y) (h'_x(y))^2 dy 
                             = 1 + \int_0^x \frac{1}{w(y)}dy
                             = h_x(x)\,.
\end{eqnarray*}
Hence, the Lemma follows.
\end{proof}

Under the additional assumption
$$\int_0^\infty\frac{1}{w(x)}dx<\infty$$
the functions in $H_w$ have a continuous continuation to infinity and the 'point evaluation at infinity' corresponds to taking the scalar product with the function
\begin{equation}
\label{def:h_infinity}
h_\infty:\mathbb R_+\rightarrow\mathbb R,x\mapsto 1 + \int_0^x \frac{1}{w(s)} ds\,.
\end{equation}
These facts are summarised in the following Lemma.
\begin{lem}\label{l:punktauswertung bei undendlich}
 Assume that $\int_0^\infty\frac{1}{w(x)}dx<\infty$. Then
 $$ \Vert g\Vert_\infty:= \sup_{x\in\mathbb R_+}\vert g(x)\vert \leq c\Vert g\Vert_w.$$
 where $c := \sqrt{1+\int_0^\infty\frac{1}{w(x)}dx}$. Moreover, $h_\infty\in H_w$ and
 $$ \lim_{x\rightarrow\infty} g(x) = \< h_\infty,g\>.$$
 for any $g\in H_w$ and $h_{\infty}$ defined in \eqref{def:h_infinity}.
\end{lem}
\begin{proof}
 We have
 $$\Vert h_\infty\Vert^2_w = 1 + \int_0^\infty \frac{1}{w(x)}dx = c^2.$$
 Hence $h_\infty\in H_w$. With $h_x$ given in Lemma \ref{l:stetigkeit der Punktauswertung} we see that
 $$ \Vert h_\infty-h_x\Vert_w^2 = \int_x^\infty \frac{1}{w(y)}dy \rightarrow 0,\quad x\rightarrow \infty.$$
 Thus $h_x\rightarrow h_\infty$ in $H_w$ for $x\rightarrow\infty$. Thus Lemma \ref{l:stetigkeit der Punktauswertung} implies
 \begin{eqnarray*}
\lim_{x\rightarrow\infty} g(x) &=& \lim_{x\rightarrow\infty}\<h_x, g\> 
 =\< h_\infty,g\>\,.
 \end{eqnarray*}
In particular, $g$ is bounded. Let  $x\in\mathbb R_+$. Then Lemma \ref{l:stetigkeit der Punktauswertung} yields
 \begin{eqnarray*}
  \vert g(x)\vert^2 &\leq& \Vert h_x\Vert^2_w\Vert g\Vert^2_w 
                  = h^2_x(x) \Vert g\Vert^2_w 
                  \leq c^2 \Vert g\Vert^2_w\,,
 \end{eqnarray*}
 and hence $\Vert g\Vert_\infty \leq c\Vert g\Vert_w$.
\end{proof}
The following technical Lemma is needed later. It shows that the square function is Lipschitz-continuous on bounded sets.
\begin{lem}\label{l:Stetigkeit der Quadratfunktion}
 Assume that $k:=\int_0^\infty\frac{1}{w(x)}dx<\infty$. Then
 $$\Vert g_1^2-g_2^2\Vert_w \leq \sqrt{5+4k^2} \Vert g_1+g_2\Vert_w \Vert g_1-g_2\Vert_w$$
 for any $g_1,g_2\in H_w$. In particular,
 $$(\cdot)^2:H_w\rightarrow H_w, g\mapsto g^2$$
 is Lipschitz-continuous on bounded subsets of $H_w$.
\end{lem}
Clearly, the second part follows from the first. Since a direct proof of the first part is very technical, we delegate it to Corollary \ref{k:Stetigkeit der Quadratfunktion} below where
the Lipschitz-continuity follows easily from the Banach algebra structure of $H_{w}$ 
relative to pointwise multiplication (see Prop.~\ref{k:Banachalgebra}).

In the definition of the dynamics of $f$, the forward price, the differential operator
with respect to $x$, $\partial_x$,  naturally occurs. We have the following convenient result for this operator on $H_w$: 
\begin{thm}\label{s:Halbgruppe von dx}
 We have
 $$ \partial_x:\{g\in H_w: g'\in H_w\}\rightarrow H_w,g\mapsto g'$$
 is the generator of the strongly continuous semigroup given by
 $$ \U_tg(x) = g(t+x)$$
 for any $t,x\in\mathbb R_+$, $g\in H_w$.
\end{thm}
\begin{proof}
 See  Filipovi\'c~\cite[Theorem 5.1.1, Remark 5.1.1]{filipovic.01}.
\end{proof}

The method of the moving frame works especially well with quasi-contractive generators of strongly continuous semigroups, cf. Tappe~\cite{Tappe.12}.
\begin{lem}\label{l:quasi-contractive}
 The semigroup $\U_t$ generated by $\partial_x$ is quasi-contractive, i.e.\ there is a constant $\beta>0$ such that
 $$ \Vert \U_t \Vert_{\op} \leq e^{t\beta}$$
 for any $t>0$.
\end{lem}
\begin{proof}
 Let $g\in H_w$ such that $\Vert g\Vert_w\leq 1$ and $t\in[0,1]$. Then
 \begin{eqnarray*}
  \Vert \U_tg\Vert^2_w &=& g^2(t) + \int_0^\infty (g'(t+y))^2 w(y)dy \\
    &\leq& g^2(0) + 2g(0)\int_0^tg'(y)dy + \left(\int_0^tg'(y)dy\right)^2 + \int_t^\infty (g'(y))^2w(y) dy \\
    &\leq& g^2(0) +2\Vert g\Vert_w \vert\<g,h_t-h_0\>\vert + t\int_0^t (g'(y))^2dy  + \int_t^\infty (g'(y))^2w(y) dy \\
    &\leq& g^2(0) +2\Vert g\Vert^2_w \Vert h_t-h_0\Vert_w + \int_0^\infty (g'(y))^2w(y)dy \\
    &=& \Vert g\Vert^2_w (1+2\Vert h_t-h_0\Vert_w)
 \end{eqnarray*}
 where $h_t$ is defined as in Lemma \ref{l:stetigkeit der Punktauswertung}. Above we used the Schwarz inequality for the second inequality and the Cauchy-Schwarz inequality for the third inequality. Thus we have
 $$ \Vert \U_t\Vert_{\op} \leq \sqrt{1+2\Vert h_t-h_0\Vert_w} \leq 1 + t \leq \exp(t)$$
 which shows the asserted inequality for $t\in[0,1]$. The inequality for $t>1$ follows by iteration.
\end{proof}
With these results on the space $H_w$ at hand, we are now ready to start our analysis
of the forward dynamics. 

\subsection{General forward curve dynamics}
We consider a general stochastic equation of the form
\begin{align}
 df(t) &= (\partial_xf(t) +\beta(t)) dt +  \Psi(t)dL(t),\quad t\geq 0\label{G:f dynamik}
\end{align}
in the sense of its mild formulation
\begin{align}\label{G:f Volterra representation}
 f(t) &= \U_tf_0 + \int_0^t\U_{t-s}\beta(s) ds +  \int_0^t\U_{t-s}\Psi(s)dL(s),\quad t\geq 0
\end{align}
to model the dynamics of the forward price specified in the Musiela parametrisation. Here, $(\U_t)_{t\geq 0}$ denotes the semigroup generated by $\partial_x$, cf.\ Lemma \ref{l:quasi-contractive}, and we assume $L$ to be a square-integrable mean-zero L\'evy process with values in some Hilbert space $U$ and $\Q$ its covariance operator.
Furthermore, we denote by $f(0)=f_0$ the initial condition, where $f_0\in H_w$.  The 
"volatility" of the forward dynamics is $\Psi\in\mathcal L^2_{L}(H_w)$, and we have the integrable drift 
$\beta:\mathbb R_+\times \Omega\rightarrow H_w$. For the space
$H_w$, we assume that $k^2:=\int_0^\infty\frac{1}{w(s)}ds<\infty$.

\begin{rem}
 If $(b,C,F)$ is the characteristics of $L$ relative to some truncation function $\chi$, then the mean $m$ and the covariance operator $\Q$ of $L(1)$ are given by 
\begin{align*}
 \Q x & = Cx + \int_U y\<x,y\> F(dy)\quad\text{and} \\
 m & = b + \int_U (y-\chi(y)) F(dy)
\end{align*}
for any $x\in U$, cf.\ \cite[Definition 4.45]{peszat.zabczyk.07}, the angle bracket is given by $\<L,L\>_t = t \tr(\Q)$ cf.\ \cite[page 35]{peszat.zabczyk.07}, the operator angle bracket by $\<\<L,L\>\>_t = t\Q$ cf.\ \cite[Theorem 8.2]{peszat.zabczyk.07} and hence the martingale covariance by $\Q/\tr(\Q)$ cf.\ \cite[Definition 8.3]{peszat.zabczyk.07}.
\end{rem}
 

A natural way to set up an equation like \eqref{G:f dynamik} is to consider stochastic partial differential equations, i.e.\ to assume that $\beta(t) = a(t,f(t))$, $\Psi(t) = \psi(t,f(t))$ for some suitable functions $a,\psi$. The following statement is about existence and
uniqueness of mild solutions to \eqref{G:f dynamik} needed for this paper under such a state-dependent specification of $\beta$ and $\psi$. 
\begin{prop}\label{p:f ODE}
Let $a:\mathbb R_+\times H_w\rightarrow H_w$, $\psi:\mathbb R_+\times H_w\rightarrow L(U,H_w)$ such that
there is an increasing function $K:\mathbb R_+\rightarrow\mathbb R_+$ such that the following Lipschitz conditions hold.
 \begin{align*}
  \Vert a(t,g)-a(t,h)\Vert_w&\ \leq K(t)\Vert g-h\Vert_w,\quad g,h\in H_w,t\in\mathbb R_+,\\
  \Vert a(t,g)\Vert_w&\ \leq K(t)(1+\Vert g\Vert_w),\quad t\in ,g\in H_w,\mathbb R_+\\
  \Vert \psi(t,g)-\psi(t,h)\Vert_{\mathrm{op}}&\ \leq K(t)\Vert g-h\Vert_w,\quad g,h\in H_w,t\in\mathbb R_+,\\
  \Vert \psi(t,g)\Vert_{\mathrm{op}}&\ \leq K(t)(1+\Vert g\Vert_w),\quad g\in H_w,t\in\mathbb R_+,
 \end{align*}
 and let $f_0\in H_w$. Then there is a unique adapted $H_w$-valued c\`adl\`ag stochastic process $f$ satisfying
\begin{align*}
 f(t) :=\U_tf_0 + \int_0^t\U_{t-s}a(s,f(s))ds + \int_0^t 
\U_{t-s}\psi(s,f(s)) dL(s),\quad t\geq0,
\end{align*}
being the mild solution to 
 $$ df(t) = (\partial_xf(t) + a(t,f(t)))dt + \psi(t,f(t))dL(t),\quad t\geq0.$$
\end{prop}
\begin{proof}
 Lemma \ref{l:quasi-contractive} states that $\U_t$ is quasi-contractive. Thus Tappe~\cite[Theorem 4.5]{Tappe.12} yields the claim. 
%
\end{proof}
We remark that the Lipschitz and growth condition listed above are considerably stronger than the conditions stated in Tappe~\cite[Theorem 4.5]{Tappe.12}, but we have restricted our attention to such coefficients $a$ and $\psi$ for simpicity only. In the rest of the paper we do not assume any functional representation for $\beta$ and
$\Psi$ as in Proposition \ref{p:f ODE}, but apply the representation 
\eqref{G:f Volterra representation} for $f$ from time to time. We refer to Peszat and Zabczyk~\cite{peszat.zabczyk.07} for a definition and further analysis of mild solutions to stochastic partial differential equations.

Before we start our discussion of the main topic let us discuss physical and risk neutral model perspectives. For simplicity assume that the SPDE \eqref{G:f dynamik} is driven by a Wiener process $W$, i.e.\
 $$ df(t) = (\partial_xf(t) +\beta(t)) dt +  \Psi(t)dW(t),\quad t\geq 0. $$
 The fundamental theorem of asset pricing states that there is an equivalent $\sigma$-martingale measure $Q\sim P$, cf. Delbaen and Schachermayer~\cite[Main theorem 1.1]{delbaen.schachermayer.98} for any market with finitely many securities. To avoid further complications we simply assume the existence of a martingale measure $Q\sim P$. Then Carmona and Tehranchi~\cite[Theorem 4.2]{carmona.tehranchi.06} yields that the $Q$-dynamics of $f$ are given by
 $$ df(t) = (\partial_xf(t) + (\beta(t)- \Q^{1/2}\psi(t))) dt +  \Psi(t)dW(t),\quad t\geq 0 $$
 for some $H_w$-valued predictable stochastic process $\psi$ where $\Q$ is the martingale covariance of $W$. However, the forwards $F(t,T)$ in delivery time parametrisation must be martingales under the martingale measure $Q$ and hence Theorem \ref{s:F dynamik} below yields
   $$ \beta(t) = \Q^{1/2}\psi(t),\quad t\geq 0.$$
 We conclude that the no-arbitrage paradigm implies that the drift term $\beta$ must stay in the range of $\Q^{1/2}$. Of course, if the driving noise is not a Wiener process, then the analysis gets more complicated which we can already see in the finite dimensional case, cf.\ Jacod and Shiryaev~\cite[Theorem III.3.24]{js.87}. If we want to model under the risk neutral measure directly, then we can and must put $\beta=0$.

 Note that we recover the spot price dynamics 
$S(t)$ 
by applying
the evaluation map $\delta_0$ on $f$ (cf.\ Proposition \ref{l:stetigkeit der Punktauswertung}):
\begin{equation}
\label{def-spot}
   S(t):= \delta_0(f(t))\,, 
\end{equation}

We have the following result for the implied spot price dynamics from the forward price
$f(t)$ in the case $L$ is a subordinated Brownian motion in $U$:
\begin{thm}\label{s:dynamik}
 If $L$ is a subordinated Brownian motion, i.e.\ $L(t) = W(\Theta(t))$, $t\geq 0$ for some $U$-valued Brownian motion $W$ and some subordinator $\Theta$, then
 \begin{align}\label{e:Stock dynamic}
   S(t) = f_0(t) + \int_0^t \beta(s)(t-s)ds + \int_0^t \sigma(t,s) dN(s)
 \end{align}
 for any $t\in\mathbb R_+$ where $N(t) = B(\Theta(t))$, $t\geq 0$ and $B$ is a constant multiple of an $\mathbb R$-valued Brownian motion and
 \begin{align}\label{e:g}
   \sigma^2(t,s) = \< \Psi(s)\Q\Psi^*(s)h_{t-s},h_{t-s} \>
 \end{align}
 for $s,t\in\mathbb R_+$. 
\end{thm}
\begin{proof}
 Note that if $B$ is an $\mathbb R$-valued Brownian motion, $\phi\in\L_L^2(\mathbb R)$ is an operator valued predictable process, then the integral defined as in Peszat and Zabczyk~\cite{peszat.zabczyk.07} can be represented as a stochastic integral as in Jacod and Shiryaev~\cite{js.87} by
 $$ \int_0^t \phi(s) dB(s) = \int_0^t \phi(s)(1) dB(s) $$
 and by the It\^o isometry we have 
 $$\E \left(\int_0^t \phi(s) dB(s)\right)^2 = \int_0^t \E\tr(\phi(s)\phi(s)^*) ds = \int_0^t \E\left((\phi(s)(1))^2\right) ds.$$
The definition of $S$ in \eqref{def-spot} and the representation of $f$ in \eqref{G:f Volterra representation} yield
 \begin{eqnarray*}
  S(t) &=& \delta_0(f(t)) 
       = f_0(t) + \int_0^t \beta(s)(t-s)ds +  \int_0^t \delta_{t-s}\Psi(s)dL(s)
 \end{eqnarray*}
 for any $t\in\mathbb R_+$. Thus, applying Thm.~\ref{S:endl. Darstellung, subordiniert} and Lemma~\ref{l:stetigkeit der Punktauswertung} we find
 $$ S(t) = f_0(t) + \int_0^t \beta(s)(t-s)ds + \int_0^t \sigma(t,s) dN(s)$$
 for any $t\in\mathbb R_+$. Here,  $N$ is an $\mathbb R$-valued subordinated Brownian motion of the same type as $L$, and
 \begin{eqnarray*}
  \sigma^2(t,s) &=& \delta_{t-s}\Psi(s)\Q\Psi^*(s)\delta^*_{t-s}1 \\
           &=& \<\delta_{t-s}\Psi(s)\Q\Psi^*(s)h_{t-s},1\> \\
           &=& \<\Psi(s)\Q\Psi^*(s)h_{t-s},h_{t-s}\>
 \end{eqnarray*}
 for any $s,t\in\mathbb R_+$.
\end{proof}

Let us for for a moment assume that $\beta=0$ in Thm.~\ref{s:dynamik}.
In the case $\Psi$ is deterministic, $\sigma(t,s)$ will become deterministic as well. Hence,
the spot price dynamics will be a Volterra process. If, additionally, $\Psi$ is constant, then
$$
\sigma^2(t,s)= \< \Psi\Q\Psi^*h_{t-s},h_{t-s} \>=:\gamma^2(t-s)\,.
$$ 
But then $S$ will be a so-called L\'evy stationary process.
An extension of L\'evy stationary processes are the so-called {\it L\'evy semistationary}
(LSS) processes.  By choosing $\Psi(s) =\widetilde{\sigma}(s) \T$ for some linear operator $\T\in L(U,H_w)$ and a predictable $\mathbb R$-valued stochastic $\widetilde{\sigma}$, then $S$ in Thm.~\ref{s:dynamik} becomes an LSS process. Indeed, the volatility is in this case
$$
\sigma^2(t,s)=\widetilde{\sigma}^2(s)\< \T\Q\T^*h_{t-s},h_{t-s} \>\,.
$$
\sloppy{ Such processes have been applied to model the stochastic dynamics of
energy prices (see Barndorff-Nielsen et al.~\cite{BNBV-spot}), but also other random phenomena
like wind speed and temperature (see Benth and \v{S}altyt\.{e} Benth~\cite{BSB-book}) and turbulence (see Barndorff-Nielsen and Schmiegel~\cite{BNSch}). }
We may discuss similar specifications of the drift $\beta$ in light of these different special classes of spot models. 

We are interested in a similar result for the dynamics of the forward price $t\mapsto F(t,T)$,
$t\leq T$, for a contract delivering at time $T$.  Obviously,
\begin{equation}
F(t,T):=\delta_{T-t}f(t)\,,
\end{equation}
for $T\geq t\geq 0$.
A natural adaption of the proof of Theorem~\ref{s:dynamik} reveals an analogous representation of $F(t,T)$ as for 
the spot price $S(t)$. We present this result for the dynamics of the bivariate process
$t\mapsto (F(t,T_1),F(t,T_2))=(\delta_{T_1-t}f(t),\delta_{T_2-t}f(t))$, with $0\leq t\leq T_1\leq T_2<\infty$. This will not only provide us with the dynamics of a single contract, but also give information about how two forward contracts with different times of delivery are connected within our model.  
\begin{thm}\label{s:F dynamik}
 If $L$ is a subordinated Brownian motion, i.e.\ $L(t) = W(\Theta(t))$, $t\geq 0$ for some $U$-valued Brownian motion $W$ and some subordinator $\Theta$, then
 \begin{eqnarray*}
   \vektor{F(t,T_1)}{F(t,T_2)} = \vektor{f_0(T_1)}{f_0(T_2)}+\int_0^t\vektor{\beta(s)(T_1-s)}{\beta(s)(T_2-s)}ds + \int_0^t \Sigma(T_1,T_2,s)dN_2(s)
 \end{eqnarray*}
 for any $t\leq T_1$. Here, $N_2(t)=(B_1(\Theta(t)),B_2(\Theta(t)))$, $t\geq 0$ where 
$B=(B_1,B_2)$ is a $\mathbb R^2$-valued standard Brownian motion and $\Sigma(T_1,T_2,s)$ is the non-negative $2\times 2$ matrix-valued process
with elements 
 $$ (\Sigma(T_1,T_2,s))^2_{i,j} = \< \Psi(s)\Q\Psi^*(s)h_{T_i-s},h_{T_j-s}\> $$
 for any $s,t\in\mathbb R_+$, $i,j = 1,2$.
\end{thm}
\begin{proof}
 This is along the same lines as the proof of Theorem \ref{s:dynamik}.
\end{proof}
We note that $\Sigma(T_1,T_2,s)_{i,i}$ will depend on $T_i$ only, for $i=1,2$. 
The dependency structure between $F(t,T_1)$ and $F(t,T_2)$ will be determined by $\Sigma(T_1,T_2,s)_{1,2}$, which we can intepret as the covariance in case of 
$L=W$ and $\Psi$ deterministic. This explicit representation of two forward
contracts can be utilized in the analysis of calendar spread options, that is, options 
with payoff depending on $F(T,T_1)-F(T,T_2)$ at some exercise time $T\leq T_1$.

\begin{ex}\label{b:Psi=I}
Consider a simple example of a forward dynamics: Let $\Psi(s)=I$, the identity operator
on $H_w$ and assume $W$ is a Wiener process with values in $U=H_w$. Furthermore,
let $\beta=0$. Hence, 
$$
f(t)=\U_tf_0+\int_0^t\U_{t-s}\,dW(s)\,.
$$
Since $\delta_y \U_s=\delta_{y+s}$, we find
\begin{align*}
F(t,T)=\delta_{T-t}f(t)&=f_0(T)+\int_0^t\delta_{T-s}\,dW(s)\,.
\end{align*}
We have that $B_T(t):=\int_0^t\delta_{T-s}dW(s)$, $t\in[0,T]$ and $T> 0$ is a time-inhomogeneous Brownian motion. Its quadratic variation structure is given by 
$\E(B_T^2(t)) = \int_0^t (\Q h_{T-s})(T-s) ds$ and $\E(B_{T_1}(t)B_{T_2}(t)) = \int_0^t(\Q h_{T_1-s})(T_2-s) ds$. Note that the process 
$t\mapsto B_t(t)=\int_0^t\delta_{t-s}dW(s)$ becomes a Gaussian martingale, and 
$$
S(t)=f_0(t)+B_t(t)
$$ 
is the spot price dynamics.
\end{ex}

In view of Theorem~\ref{s:F dynamik} we would like to analyse in more detail the spatial correlation structure for random fields defined as stochastic processes with values in 
the Hilbert space $H_w$.  
For this purpose we show first a technical Lemma, which reveals that two point evaluations on $H_w$ at close points are almost the same and hence cannot be orthogonal at all. In other words, the specific geometric structure of the space $H_w$ implies a strong connection between the different point evaluations.
\begin{lem}\label{l:hoelderstetigkeit der Punktauswertung}
Recall the function $h_x:\mathbb R_+\rightarrow\mathbb R,z\mapsto 1 + \int_0^{x\wedge z}\frac{1}{w(s)}ds$ for any $x\in\mathbb R_+$ defined in Lemma~\ref{l:stetigkeit der Punktauswertung}. Then $$h:\mathbb R_+\rightarrow H_w, x\mapsto h_x$$
 is $0.5$-H\"older continuous with constant $1$. Moreover,
 \begin{align*}
  \Vert h_y - h_x\Vert_w &\geq \sqrt{\frac{y-x}{w(y)}}
 \end{align*}
  for any $0\leq x\leq y<\infty$.  
\end{lem}
\begin{proof}
 Let $0\leq x\leq y<\infty$. Then
\begin{align*}
  \Vert h_y-h_x\Vert_w^2 &\ = \int_x^y\frac{1}{w(s)}ds \begin{cases}\leq y-x&\text{and}\\ \geq \frac{y-x}{w(y)}.& \end{cases}
\end{align*}
and hence $\Vert h_y-h_x\Vert_w \leq \sqrt{y-x}$ and 
$\Vert h_y-h_x\Vert_w \geq \sqrt{(y-x)/w(y)}$. 
\end{proof}
Note that for $0\leq x\leq y<\infty$ we have $\<h_y,h_x\> = \Vert h_x\Vert_w^2$, and therefore it also holds that
$$
 \frac{\<h_y,h_x\>}{\Vert h_y\Vert_w\Vert h_x\Vert_w} = \frac{\Vert h_x\Vert_w}{\Vert h_y\Vert_w}\,. 
$$

Consider now a square integrable random variable $X$ with values in $H_w$ having mean zero and covariance operator $\Q$. The correlation between $X(x):=\delta_x(X)$ and
$X(y)=\delta_y(X)$ is given as
\begin{equation}
\label{def:correl}
\rho:\mathbb R_+^2\rightarrow\mathbb R, (x,y)\mapsto \begin{cases} \frac{\E\left[\delta_x(X)\delta_y(X)\right]}{\sqrt{\E\left[\delta_x(X)^2\right]\E\left[\delta_y(X)^2\right]}} &\text{if the denominator is non-zero,}\\1&\text{otherwise}.\end{cases}
\end{equation}
The next Theorem shows, based on the Lemma above, that the correlation 
converges to one at a rate $\sqrt{|x-y|}$ as $x\rightarrow y$:
\begin{thm}\label{S:lokale korrelation}
Let $X$ be a mean zero, square integrable and $H_w$-valued random variable with covariance operator $\Q$. For any $x\in\mathbb R_+$ there is $\epsilon >0$ such that for $\rho$ defined in
\eqref{def:correl} it holds 
 $$ \rho(x,y) \geq 1 - \frac{2\Vert \Q\Vert_{\mathrm{op}}^{1/2}\sqrt{|x-y|}}{\Vert \Q^{1/2}h_x\Vert_w+\Vert \Q\Vert_{\mathrm{op}}^{1/2}\sqrt{|x-y|}}$$
 for any $|x-y|\leq\epsilon$.
\end{thm}
\begin{proof}
Denote by $C:=\Q^{1/2}$ the positive semidefinite square root of $Q$ and let $x\in\mathbb R_+$. If $\Vert \Q^{1/2}h_x\Vert_w=0$, then the claimed inequality is trivial. Thus we assume that $\Vert \Q^{1/2}h_x\Vert_w\neq0$. 

Define
 $$ \epsilon := \frac{\Vert Ch_x\Vert_w^2}{\Vert \Q\Vert_{\mathrm{op}}}$$
 and let $y\in\mathbb R_+$ such that $\vert y-x\vert <\epsilon$. Then Lemma \ref{l:hoelderstetigkeit der Punktauswertung} yields
 $$ \Vert Ch_y\Vert_w \geq \Vert Ch_x\Vert_w-\Vert C\Vert_{\mathrm{op}}\sqrt{\vert y-x\vert} > 0$$
 and hence 
$$\rho(x,y) = \frac{\E\left[\delta_x(X)\delta_y(X)\right]}{\Vert Ch_x\Vert_w\Vert Ch_y\Vert_w}.$$
We have
 $$ \E\left(\delta_x(X)\delta_y(X)\right) = \<\Q h_x,h_y\>$$
 by Lemma \ref{l:stetigkeit der Punktauswertung}. Moreover, we have
 \begin{align*}
   \E\left[\delta_x(X)^2\right] & = \<\Q h_x,h_x\> = \Vert Ch_x\Vert_w^2,
\end{align*}
and
\begin{align*}
   \Vert Ch_y\Vert_w & \leq \Vert Ch_x\Vert_w + \Vert C\Vert_{\mathrm{op}}\Vert h_y-h_x\Vert_w \\
  & \leq \Vert Ch_x\Vert_w + \Vert C\Vert_{\mathrm{op}}\sqrt{\vert y-x\vert}. 
 \end{align*}
Thus we get
\begin{align*}
 \E\left[\delta_x(X)\delta_y(X)\right] & = \<Ch_x,Ch_y\> \\
    & \geq \<Ch_x,Ch_x\>-\vert\<Ch_x,C(h_y-h_x)\>\vert \\
    & \geq \Vert Ch_x\Vert_w^2 - \Vert Ch_x\Vert_w \Vert C\Vert_{\mathrm{op}}\Vert h_y-h_x\Vert_w \\
    & \geq \Vert Ch_x\Vert_w\left(\Vert Ch_x\Vert_w - \Vert C\Vert_{\mathrm{op}}\sqrt{\vert y-x\vert}\right)
\end{align*}
where the third inequality follows from the Cauchy-Schwarz inequality. We conclude
\begin{align*}
  \rho(x,y) & \geq \frac{\Vert Ch_x\Vert_w - \Vert C\Vert_{\mathrm{op}}\sqrt{\vert y-x\vert}}{\Vert Ch_x\Vert_w + \Vert C\Vert_{\mathrm{op}}\sqrt{\vert y-x\vert}} \\
                 & = 1 - \frac{2\Vert \Q\Vert_{\mathrm{op}}^{1/2}\sqrt{\vert y-x\vert}}{\Vert \Q^{1/2}h_x\Vert_w+\Vert \Q\Vert_{\mathrm{op}}^{1/2}\sqrt{\vert y-x\vert}}
\end{align*}
 because $\Vert \Q\Vert_{\mathrm{op}}^{1/2} = \Vert \Q^{1/2}\Vert_{\mathrm{op}}$ by Palmer~\cite[Theorem 3.4.21]{palmer.94}.
\end{proof}

Observe, that there is no a priori upper bound for the correlation. Indeed, consider the random variable $X:=(h_x+h_y)A$ where $A$ is a standard normal random variable, $0\leq x\leq y\leq \infty$. Then, $\rho(x,y) = 1$.

Thm.~\ref{S:lokale korrelation} implies that the correlation among forwards with close maturities is high, which is very natural from a practical perspective as there is little
economical difference of having a commodity delivered at different but nearby times. 
However, it is 
remarkable that the spatial geometry of the space $H_w$ is imposing an explicit lower
bound for the correlation, which is stationary in distance between times to delivery and
following a square-root function.

\subsection{Representing the forwards by a sum of Ornstein-Uhlenbeck type processes}

In commodity markets a popular class of spot price models is given in terms of a
finite sum of Ornstein-Uhlenbeck processes (see for instance Benth et al.~\cite{BSBK-book} 
and the references therein).  In the traditional set-up, one is modelling the logarithmic spot
price dynamics by a finite sum of Ornstein-Uhlenbeck processes, each driven by a
L\'evy process. However, several papers also advocate to model the spot price
dynamics directly by such a finite sum of Ornstein-Uhlenbeck processes (see for example
Benth, Meyer-Brandis and Kallsen~\cite{benth.et.al.05} and 
Garcia, Kl\"uppelberg and M\"uller~\cite{GKM} for the case of power). In a very simplistic setting, one starts out with 
a model
$$
S(t)=\Lambda(t)+X(t)\,,
$$  
where $\Lambda$ is some deterministic seasonality function, and $X$ follows an
Ornstein-Uhlenbeck process
$$
dX(t)=-\lambda X(t)\,dt+dN(t)
$$
for some $\mathbb R$-valued L\'evy process $N$ and $\lambda>0$ a constant. One defines the forward
price to be the conditional expectation of the spot price at delivery, given information at current
time, where the expectation is taken under some pricing measure 
(see Benth et al.~\cite{BSBK-book}). I.e., 
$$
F(t,T)=\E[S(T)\,|\,\mathcal{F}_t]
$$
for $t\leq T$. To avoid unnecessary technical details, we assume that the spot model is already
stated under the pricing measure, and that the L\'evy process $N$ is integrable. Then it is a straightforward task to derive that 
$$
F(t,T)=\Lambda(T)+\frac{1}{\lambda}\E[N(1)](1-e^{-\lambda(T-t)})+
e^{-\lambda(T-t)}X(t)\,.
$$ 
In other words, the forward price is linear in the Ornstein-Uhlenbeck process, with a coefficient
which is exponentially decreasing in time to maturity $T-t$ at a rate given by the 
speed of mean reversion $\lambda$. If we 
consider a spot model being a sum of such Ornstein-Uhlenbeck processes, we find that the above expression for the forward price will generalize naturally to a sum of exponentially weighted Ornstein-Uhlenbeck processes.

We want now to analyse to what extent the opposite holds, that is, when can we represent our 
forward price dynamics $f(t)$ being a process with values in $H_w$ as a weighted series of Ornstein-Uhlenbeck processes? Of course, as long as we are letting the forward price be perturbed by an infinite dimensional noise we cannot expect a {\it finite} sum of Ornstein-Uhlenbeck processes (appropriately weighted) representing the stochastic evolution of the curve in $H_w$, but an {\it infinite} representation may be within reach. The main technical obstacle for such an infinite series representation is that the exponential functions do not constitute any orthogonal set in $H_w$. However, if the forward curve takes 
values in a convenient subspace of $H_w$ having exponential functions as a Riesz basis, 
it turns out that we can find a series representation in terms of Ornstein-Uhlenbeck 
processes. We present some auxiliary results to show this.

Recall that a Riesz basis is a sequence $(x_n)_{n\in\mathbb N}$ in a Hilbert space $H$ such that there is an orthonormal basis $(e_n)_{n\in\mathbb N}$ in $H$ and an invertible linear operator $\T$ such that $\T e_n = x_n$. For further equivalent statements we refer the reader to Young~\cite[Theorem 1.9]{young.80}.
In the next Theorem we find a Riesz basis consisting only of exponential functions for a 'large' subspace of $H_w$. This subspace will contain a natural copy of the first Sobolev space over $L^2([0,T],\mathbb C)$. We formulate our results for the particular choice of weight functions $w:x\mapsto e^{\alpha x}$ for $\alpha>0$ in $H_w$:
\begin{thm}\label{s:exponentielle Riesz basis}
 Let $\lambda>0$, and fix $x_0>0$. 
Then there is a closed subspace $H_w^{x_0}$ such that the following statements hold.
 \begin{enumerate}
  \item $H_w^{x_0}$ has a Riesz basis $(g_n)_{n\in \mathbb Z}$ such that $g_0(x)=1$, $x\in\mathbb R_+$ and $g_n(x) = \frac{1}{\lambda_n\sqrt{x_0}}(1-e^{\lambda_nx})$, $x\in\mathbb R_+,n\neq 1$ where $\lambda_n=\frac{2\pi i n}{x_0}-\lambda-\alpha/2$.
  \item $\overline g_n = g_{-n}$ for any $n\in\mathbb C$.
  \item There is a continuous linear projection $\Pi_{x_0}:H_w\rightarrow H_w^{x_0}$ such that $\Pi_{x_0} g(x) = g(x)$ for any $g\in H_w$, $x\in[0,x_0]$.
  \item $H_w^{x_0}$ is invariant under the semigroup $(\U_t)_{t\geq 0}$ defined in 
Thm.~\ref{s:Halbgruppe von dx}.
  \item If $(g_n^*)_{n\in \mathbb Z}$ is the corresponding biorthogonal system in the sense of Young~\cite[page 28]{young.80} and 
$(\U_t)_{t\geq 0}$ the semigroup defined in Thm.~\ref{s:Halbgruppe von dx}, then $\U_t^*g_n^*=e^{-\lambda_nt}g_n^*$ and $g_0^*=g_0$.
 \end{enumerate}
\end{thm}
\begin{proof}
Let $V$ be as in Lemma \ref{l:Riesz basis}. Then Lemma~\ref{l:Riesz basis} states that $(\tilde g_n)_{n\in\mathbb Z}$ is a Riesz basis of $V$ where
$$\tilde g_n:\mathbb R_+\rightarrow\mathbb C,x\mapsto \frac{1}{\sqrt{x_0}}e^{(\frac{2\pi i n}{x_0}-\lambda) x},\quad n\in\mathbb N.$$
Define $H_w^{x_0}:=\{f\in H_w:f'\sqrt{w}\in V\}$ and
\begin{align*}
 g_n(x) &:= \int_0^x \tilde g_n(y)e^{-y\alpha/2}dy = \frac{1-e^{\lambda_n x}}{\lambda_n\sqrt{x_0}},\quad x\in\mathbb R_+.\\
\Psi: &L^2(\mathbb R_+,\mathbb C)\rightarrow H_w,f\mapsto\left(x\mapsto \int_0^x f(y)e^{-y\alpha/2}dy\right) 
\end{align*}
is an isometric embedding and hence (1) and (2) follow. For the remainder of the proof let $(g_n)_{n\in\mathbb Z}$ be as in (1).

Let $F:=\{f\in H_w:f(x)=0,x\in[0,x_0]\}$. Then $F$ is a closed vector space complement to $H_w^{x_0}$ in $H_w$. Hence there is a continous linear projection $\Pi:H_w\rightarrow H_w^{x_0}$ with kernel $F$. Let $g\in H_w$, $x\in[0,x_0]$. Then $g-\Pi g$ is in the kernel of $\Pi$. Hence $g-\Pi g\in F$ which implies (3).

Let $(\U_t)_{t\geq0}$ be the shift semigroup defined in Thm.~\ref{s:Halbgruppe von dx}. Then
 $$ \U_t g_n = e^{-\lambda_n t}g_n + g_n(t)g_0 \in H_w^{x_0}\,,$$
which shows (4).

Since $g_0$ is normed and orthogonal to $(g_n)_{n\neq 0}$ we have $g_0^*=g_0$. Young~\cite[Theorem 1.8]{young.80} implies that $(g_n^*)_{n\in\mathbb Z}$ is a Riesz basis as well. Let $n\neq 0$. Then we have
\begin{align*}
 \U_t^*g_n^* &= \sum_{k\in\mathbb Z} \<\U_t^*g_n^*,g_k\> g_k^* \\
            &= \sum_{k\in\mathbb Z} \<g_n^*,\U_tg_k\> g_k^* \\
            &= \sum_{k\in\mathbb Z} \<g_n^*,e^{-\lambda_k t}g_k+g_k(t)g_0\> g_k^* \\
            &= e^{-\lambda_n t}g_n^*.
\end{align*}
Hence, the proof is complete.
\end{proof}

In view of Thm.~\ref{s:exponentielle Riesz basis} any stochastic process $X$ on $H_w$ can be mimicked by a stochastic process $Y$ on $H_w^{x_0}$, i.e.\ $Y(t,x)=X(t,x)$ for any $x\in[0,x_0]$. If we assume that the forwards price process $f(t)$ evolves in the space $H_w^{x_0}$, then it can be represented by a sum of Ornstein-Uhlenbeck type processes.
\begin{thm}
 Let $x_0>0$ and $H_w^{x_0}$ be as in Theorem~\ref{s:exponentielle Riesz basis}. Assume that $f(t)$ is $H_w^{x_0}$-valued. Then there is a sequence $(M_n)_{n\in\mathbb N}$ of complex valued square integrable martingales such that
  $$f(t) = S(t) + 2\sum_{n=1}^\infty \Re\left(g_n\int_0^t e^{(s-t)\lambda_n}\left\{\<g^*_n,\beta(s)\>\,ds+dM_n(s)\right\}\right),\quad t\in\mathbb R_+.$$
 where the sum converges almost surely in $H_w$ and $(g_n)_{n\in\mathbb Z}$ is the Riesz basis provided in Thm.~\ref{s:exponentielle Riesz basis}.

 Moreover, if $\Psi$ is deterministic, then $M_n$ is a process with independent increments. If $\Psi$ is deterministic and constant, then $M_n$ is a L\'evy process. If $\Psi$ is deterministic and $L$ a Brownian motion, then $(M_n)_{n\geq 1}$ is a sequence of Gaussian processes.
\end{thm}
\begin{proof}
  Let $(g_n^*)_{n\in\mathbb Z}$ be the biorthogonal system associated with $(g_n)_{n\in\mathbb Z}$. We define $M_n(t) : =\<g_n^*,\int_0^t\Psi_sdL(s)\>$ for any $n\in\mathbb Z$. Observe that $M_n$, $n\geq 1$ have the properties stated at the end of the theorem. We have
  \begin{align*}
    f(t) &= \sum_{n\in\mathbb Z} g_n\<f(t),g_n^*\> \\
        &=\<f(t),g^*_0\> + \sum_{n\in\mathbb Z,n\neq0}\left( g_n\int_0^t \<g_n^*,\U_{t-s}\beta(s) ds\>
          + g_n\int_0^t \<g_n^*,\U_{t-s}\Psi(s) dL(s)\>\right)
  \end{align*}
  for any $t\in\mathbb R_+$. As $g_0^*=1$, we find $\<f(t),g_0^*\>= S(t)$. Moreover, we have
  \begin{align*}
     \int_0^t \<g_n^*,\U_{t-s}\Psi(s) dL(s)\> &= \int_0^t\<\U_{t-s}^*g_n^*,\Psi(s)dL(s)\> \\
                      &= \int_0^t e^{(s-t)\lambda_n} \<g_n^*,\Psi(s)dL(s)\> \\
                      &= \int_0^t e^{(s-t)\lambda_n} dM_n(s)\,.
  \end{align*}
Here we have applied Thm.~\ref{s:exponentielle Riesz basis}, part (5) in the second equality. 
Similarly for the drift part $\beta$ we calculate,
  \begin{align*}
     \int_0^t \<g_n^*,\U_{t-s}\beta(s)\,ds\> &= \int_0^t\<\U_{t-s}^*g_n^*,\beta(s)\,ds\> \\
                      &= \int_0^t e^{(s-t)\lambda_n} \<g_n^*,\beta(s)\,ds\> \\
                      &= \int_0^t e^{(s-t)\lambda_n} \<g_n^*,\beta(s)\>\,ds
  \end{align*}
 for any $n\in\mathbb Z$, $n\neq0$, $t\in\mathbb R_+$. Finally, observe that $g_n\int_0^t e^{(s-t)\lambda_n} dM_n(s)$ is the complex conjugate of $g_{-n}\int_0^t e^{(s-t)\lambda_{-n}} dM_{-n}(s)$ and hence their sum equals $2\Re\left(g_n\int_0^t e^{(s-t)\lambda_n} dM_n(s)\right)$.
\end{proof}

The Theorem tells us that the forward curves are indeed representable as 
an infinite series of (complex-valued) Ornstein-Uhlenbeck processes. We must restrict
our attention to the space $H^{x_0}_w$, however, as we have from above, any element
in $H_w$ can be projected to $H_w^{x_0}$, and the two curves will coincide on $[0,x_0]$.
We may view $x_0$ as the maximal horizon of time-to-maturities of interest in the market. For example, in power markets typically contracts has deliveries for up to 4 years 
(this is the case in the Nordic market NordPool, or the German EEX market). One may model
the forward curve in $H_w$, project it down to $H_w^{x_0}$, and then have the 
representation in terms of Ornstein-Uhlenbeck processes. For the times of deliveries
in question, i.e., for $x\leq x_0$, we will have that the representation coincides
with the dynamics of $f(t)$. Outside, for $x>x_0$, we do not know if this is true
for the curve in $H_w$.

Let us give one final remark on the specific weight function $w$. The case when $w$ is not an exponential function is very delicate and requires an extensive generalization of the 
analysis above. Since we do not want to deviate from our main topic too much, we decided not to include the more general cases here.

\subsection{Factorial models and covariance representation}\label{a:faktorielle Modelle}

As we recall from Peszat and Zab\-czyk~\cite{peszat.zabczyk.07}, $L$ can be written as a sum of 
orthogonal and uncorrelated real-valued L\'evy processes based on the spectral decomposition of the covariance operator $\Q$. 
This decomposition can be used to express the forward dynamics $f(t)$ in terms of a series representation. However, in many practical applications it is difficult to find the spectral decomposition of $\Q$
explicitly. Moreover, for modelling purposes it is sometimes more convenient to specify the directions in which the driving noise pushes the forward curve. For instance, one would like to specify the covariance by
 \begin{align*}
   \Q g = \sum_{n=1}^k \<g_n,g\>g_n
 \end{align*}
where $g_1,\dots,g_k$ is a finite set of functions in $H_w$. In an empirical context,
one may view these functions as factor loadings for observed directions. Note that the representation of $\Q$ in terms of the finite set of functions $\{g_n\}_{n=1}^k$ yields that the noise $L$ can be viewed as a $k$-dimensional
L\'evy process. This can be generalised to infinite sums as long as the extra condition $\sum_{n\in\mathbb N}\Vert g_n\Vert^2_w<\infty$ is satisfied. This section focuses on such decompositions and summarises our main results in this regard.


\begin{rem}
 Let $(L_n)_{n\in I}$, $I\subseteq\mathbb N$ be a family of uncorrelated, mean-zero, square integrable L\'evy processes with variance $1$ and $(y_n)_{n\in I}$ a family in a Hilbert space $H$ such that $\sum_{n\in I}\vert y_n\vert^2 <\infty$. Then $$\sum_{j=1}^NL_n(t)y_n,\quad t\in\mathbb R_+$$
 converges in $L^2(\Omega, H)$ to some $H$-valued mean-zero, square integrable L\'evy process $L$ with covariance
  $$ \Q x = \sum_{n\in I} \<y_n,x\>y_n,\quad x\in H.$$
  If the sequence is stochastically independent, then Peszat and Zabczyk~\cite[Corollary 3.12]{peszat.zabczyk.07} yield that the convergence is $P$-a.s.
\end{rem}


Although we are working in the particular Hilbert space $H_w$, the following result
holds for general Hilbert spaces $H$ and we formulate it in such a case: 
\begin{thm}\label{s:allgemeine Darstellung}
Let $L$ be a square integrable L\'evy process with covariance $\Q$ and $m:=\E L(1)$. Let $(y_n)_{n\in I}$ be a Riesz basis of $H$ with $I=\mathbb N$ or let $(y_n)_{n\in I}$ be a finite linear independent set of elements in $H$. Assume that the covariance $\Q$ of $L$ is given by
 \begin{align}
   \Q x &= \sum_{n\in I} \< x,y_n\> y_n\quad x\in H,
  \end{align}
where
\begin{align}
\sum_{n\in I} \Vert y_n\Vert^2 &<\infty.
 \end{align}
Then there is a family $(L_n)_{n\in I}$ of uncorrelated, mean-zero, square integrable and $\mathbb R$-valued L\'evy processes with $\E L^2_n(1) = 1$ such that
\begin{itemize}
 \item $L_n$ is adapted to the filtration generated by $L$ for any $n\in \mathbb N$,
 \item $L(t) = tm + \sum_{n\in \mathbb N} L_n(t) y_n$ where the sum converges in $L^2(\Omega,H)$ uniformly in $t$ on compact intervalls.
\end{itemize}
 In particular, if $W=L$ is a Wiener process, then $W_n:=L_n$ defines a family of independent standard Brownian motions and $W(t) = \sum_{n\in \mathbb N}W_n(t)y_n$ where the convergence is $P$-a.s.\ for any $t\in\mathbb R_+$.
\end{thm}
\begin{proof}
 We may assume that $m=0$ because we can work with $\widetilde L(t):=L(t)-tm$ instead.

 For any $n\in I$ let $z_n\in H$ such that $\<z_n,y_k\>=1_{\{n\neq k\}}$ for any $k\in I$. Define
  $$ L_n(t) := \<z_n, L(t)\>,\quad t\in\mathbb R_+ $$
 Let $\Pi$ be the orthonormal projection onto the closed subspace generated by $(y_n)_{n\in I}$. We have
  $$ \Pi(L(t)) = \sum_{n\in I} L_n(t)y_n\quad\text{surely in }H $$
 for any $t\in\mathbb R_+$. Moreover, we have
 \begin{align*}
  \E(L_n(t)L_k(t)) &= t\< \Q z_n, z_k \> \\
                   &= t\< y_n,z_k\> \\
                   &= t1_{\{n=k\}}
 \end{align*}
for any $n,k\in I$. Thus $(L_n)_{n\in I}$ is a a family of uncorrelated, mean-zero, square integrable and $\mathbb R$-valued L\'evy processes with $\E L^2_n(1) = 1$. Moreover, observe that $\Q\Pi = \Q$ and hence we have
 \begin{align*}
  \E \<\Pi(L(t))-L(t) ,x\>^2 &= t\Vert \Q(\Pi x-x)\Vert^2 \\
           &= 0
 \end{align*}
for any $x\in H$.
\end{proof}

For the remainder of this Section~\ref{a:faktorielle Modelle}, we do assume the following {\em standing assumption}
\begin{align}\label{e:gv}
  \Psi(t)=\sigma(t)\T 
\end{align}
where $\sigma$ is some $\mathbb R$-valued, locally bounded and adapted 
stochastic process and $\T$ is a linear operator from $U$ to $H_w$. Then
 $\T \Q \T^*$ is a positive semidefinite trace-class operator. Hence, $\T \Q\T^*$ has a unique positive semidefinite root $C$, which is a Hilbert-Schmidt operator. Note that this specification of $\Psi$ is a simple approach 
to include seasonality and stochastic volatility into the forward curve evolution, while  
$\T$ embeds our L\'evy process into our curve space $H_w$. For example, Barndorff-Nielsen et al.~\cite{BNBV-spot} and Benth~\cite{B-sv} find evidence for stochastic volatility in power and gas spot prices, resp. Seasonality of volatility in forward prices is discussed from an empirical point of view for the NordPool market in Benth and Koekebakker~\cite{BK}, Andresen et al.~\cite{andresen.et.al.10} and Frestad et al.~\cite{benth.et.al.10}, where their results may also indicate presence of stochastic volatility. We have the following series representation of $f(t)$:


\begin{thm}\label{s:Levy Treiber}
 Let $(g_n)_{n\in I}$ be elements of $H_w$, satisfying either, 
\begin{itemize}
 \item $I\subseteq\mathbb N$ is finite and $(g_n)_{n\in I}$ is linear independent, or
 \item $I=\mathbb N$ and $(g_n)_{n\in I}$ is a Riesz-basis of a closed subspace of $H_w$ with $\sum_{n\in I}\Vert g_n\Vert_w^2<\infty$.
\end{itemize}
 Moreover, assume that $(\T \Q \T^*)g= \sum_{n\in I}\<g_n,g\>g_n$ for any $g\in H_w$. Then there is a family of mean-zero $\mathbb R$-valued uncorrelated L\'evy processes $(L_n)_{n\in I}$ with $\E (L_n(1))^ 2 = 1$ such that 
 \begin{align*}
   f(t) &=\U_tf_0 + \int_0^t\U_{t-s}\beta(s)ds + \sum_{n\in I}\int_0^t \sigma(s) \U_{t-s}g_ndL_n(s)
 \end{align*}
for any $t\geq 0$ where the sums converge in $L^2(\Omega,H_w)$. Moreover, if additionally $g_n\in \dom(\partial_x)$ for any $n\in I$, then
 \begin{align*}
   f(t) &= \U_tf_0 + \int_0^t\U_{t-s}\beta(s)ds \\
        &\ + \sum_{n\in I}\left(\int_0^t\int_0^s \left(\U_{s-r}g'_n\right)\sigma(r)dL_n(r)ds + g_n\int_0^t\sigma(s)dL_n(s) \right)
 \end{align*}
 for any $t\in\mathbb R_+$ where the sums converge in $L^2(\Omega,H_w)$. Moreover, under the additional assumption we have
 \begin{align*}
   S(t) &= f_0(t) + \int_0^t\mathcal \delta_{t-s}\beta(s)ds + \sum_{n\in I} \int_0^t g_n(t-s)\sigma(s)dL_n(s) \\
        &=f_0(t) + \int_0^t\U_{t-s}\beta(s)ds\\&\ + \sum_{n\in I} \left(\int_0^t\int_0^s g'_n(s-r)\sigma(r)dL_n(r)ds + g_n(0)\int_0^t\sigma(s)dL_n(s) \right)
 \end{align*}
 for any $t\in\mathbb R_+$.
\end{thm}
\begin{proof} The representations for $S$ simply follow from applying the continuous linear functional $\delta_0$ to the representations of $f$. 

Theorem \ref{s:allgemeine Darstellung} applied to the L\'evy process 
$\T(L(t))_{t\geq 0}$ whose covariance is $\T \Q\T^*$ yields a family of mean-zero $\mathbb R$-valued uncorrelated L\'evy processes $(L_n)_{n\in I}$ with $\E (L^2_n(1)) = 1$ such that
$$ \T (L(t)) = \sum_{n\in I} g_n dL_n(t). $$
Thus we have
\begin{align*}
 \int_0^t \U_{t-s}\Psi(s)dL(s) &= \int_0^t \sum_{n\in I} (\U_{t-s}g_n) \sigma(s) dL_n(s) \\
                              &= \sum_{n\in I}\int_0^t (\U_{t-s}g_n) \sigma(s) dL_n(s)
\end{align*}
for any $t\geq 0$. This yields the first representation for $f$. Under the additional assumption we have
 \begin{align*}
  &\ \int_0^t \U_{t-s}g_n \sigma(s)dL_n(s) 
     = \int_0^t\int_0^s \U_{s-r}g'_n\sigma(r)dL_n(r)ds + g_n\int_0^t\sigma(s)dL_n(s)\,,
 \end{align*}
which finalises the proof.
\end{proof}
Exponential curves are of course of particular interest, i.e.\ curves of the form $g_n(x)=\lambda_ne^{-\gamma_nx}$. In this case one can calculate the integrals more explicitly and it turns out that the forward curve $f(t)$ corresponds to a sum of L\'evy-driven Ornstein-Uhlenbeck type processes:
\begin{cor}
\label{cor:OU-repr}
 Let $I=\{1,\dots,d\}$, $d\in\mathbb N$, $(\gamma_n)_{n\in I}$ a family of pairwise different elements in $\mathbb R_+$ and $(\lambda_n)_{n\in I}$ be a family in $\mathbb R$, $g_n(x):=\lambda_ne^{-\gamma_nx}$ for $n\in I$, $x\in\mathbb R_+$ and assume that 
\begin{itemize}
 \item $g_n\in H_w$, and hence $g_n$ is in the domain of $\partial_x$, for any $n\in I$,
 \item $\T\Q\T^*g = \sum_{n\in I}\<g_n,g\>g_n,\quad g\in H_w$ and
 \item $\beta$ has values in the vector space generated by $g_n$, $n\in I$.
\end{itemize}
Then $$ f(t) = \U_tf_0 + \sum_{n\in I} g_n X_n(t)$$
 where 
 $$dX_n(t) = (\mu_n(t)-\gamma_nX_n(t))dt +\sigma(t)dL_n(t).$$
 for any $t\in\mathbb R_+$ where $L_n$ is as in Theorem \ref{s:Levy Treiber} and $\mu_n:\mathbb R_+\times\Omega\rightarrow\mathbb R$ is some predictable process. Moreover, we have
 $$ S(t) = f_0(t) + \sum_{n\in I} Y_n(t)$$
 where $dY_n(t) = (\mu_n(t)-\gamma_nY_n(t))dt + \lambda_n\sigma(t)dL_n(t)$ for any $t\in\mathbb R_+$.
\end{cor}
\begin{proof}
 Since $g_n'=-\lambda_ng_n\in H_w$ we have $g_n$ is in the domain of $\partial_x$. Since $(\gamma_n)_{n\in I}$ are pairwise different we have that $(g_n)_{n\in I}$ is linearly independent. Let $V\subseteq H_w$ be the subspace generated by $g_n$, $n\in I$ and $(g_n^*)_{n\in I}$ be the biorthogonal system of $(g_n)_{n\in I}$ in $V$. Define
 $$ \mu_n(t) := \<g_n^*,\beta(t)\>,\quad t\in\mathbb R_+.$$
Then we have $\beta(t) = \sum_{n\in I} g_n\mu_n(t)$, $t\geq 0$. Furthermore,
$$
\mathcal{U}_{t-s}g_n(x)=\lambda_n e^{-\gamma_n(x+t-s)}=g_n(x) e^{-\gamma_n(t-s)}\,,
$$
which yields, using Theorem~\ref{s:Levy Treiber},
\begin{align*}
f(t)&=\U_t f_0+\int_0^t\U_{t-s}\beta(s)\,ds+
\sum_{n\in I}\int_0^t\sigma(s)\mathcal{U}_{t-s}g_n\,dL_n(s) \\
&=\U_t f_0+\int_0^t\U_{t-s}\sum_{n\in I}g_n\mu_n(s)\,ds+\sum_{n\in I}g_n\int_0^t
\sigma(s) e^{-\gamma_n(t-s)}\,dL_n(s)\,\\
&= \U_t f_0+\sum_{n\in I}g_n\left(\int_0^te^{-\gamma_n(t-s)}\mu_n(s)\,ds+\int_0^t
\sigma(s) e^{-\gamma_n(t-s)}\,dL_n(s)\right)\\
&= \U_t f_0+\sum_{n\in I}g_n\left(\int_0^te^{-\gamma_n(t-s)}\big(\mu_n(s)\,ds+\sigma(s)dL_n(s)\big)\right).
\end{align*}
Define the real-valued process
$$
X_n(t)=\int_0^te^{-\gamma_n(t-s)}\big(\mu_n(s)\,ds+\sigma(s)dL_n(s)\big)\,,
$$
which, by integration-by-parts, is an Ornstein-Uhlenbeck process
$$
dX_n(t)=(\mu_n(t)-\gamma_n X_n(t))\,dt+\sigma(t)\,dL_n(t)\,.
$$
The representation for $S$ follows simply from \eqref{def-spot} by applying $\delta_0$ to the representation of $f$.
\end{proof}
Note that the condition $g_n\in H_w$ is of course equivalent to 
$$
\int_0^{\infty}w(x) e^{-2\gamma_n x}\,dx<\infty\,.
$$
In the typical case where $w$ is an exponential function $w(x)=\exp(\alpha x)$ with $\alpha>0$, 
this condition is satisfied if and only if the coefficients $\gamma_n$ are strictly bigger than $\alpha/2$. One possibility to ensure the second condition is to define the driving noise $L(t) :=\sum_{n=1}^d g_n B_n(t)$ where $B$ is some $\mathbb R^d$-valued standard Brownian motion. In that specific setup we have
\begin{cor}\label{C:brownian OU}
  Let $I=\{1,\dots,d\}$, $d\in\mathbb N$, $\alpha>0$, $w(x):=e^{\alpha x}$ for any $x\in\mathbb R_+$, $(\gamma_n)_{n\in I}$ a family of pairwise different elements in $\mathbb R_+$ with $\gamma_n>\alpha/2$ and $(\lambda_n)_{n\in I}$ be a family in $\mathbb R$, $g_n(x):=\lambda_ne^{-\gamma_nx}$ for $n\in I$, $x\in\mathbb R_+$ and assume that 
\begin{itemize}
 \item $L(t) = W(t) := \sum_{n=1}^d g_nB_n(t)$,
 \item $\T$ is the identity operator on $H_w$,
 \item the market $(F(t,T))_{0\leq t\leq T}$ does not allow for arbitrage in the sense of Delbaen and Schachermayer~\cite[Main theorem 1.1]{delbaen.schachermayer.98}
\end{itemize}
where $B=(B_1,\ldots,B_d)$ is an $\mathbb R^d$-valued standard Brownian motion.
Then $$ f(t) = \U_tf_0 + \sum_{n\in I} g_n X_n(t)$$
 where 
 $$dX_n(t) = (\mu_n(t)-\gamma_nX_n(t))dt +\sigma(t)dB_n(t)$$
 and $\mu_n:\mathbb R_+\times\Omega\rightarrow\mathbb R$ is some predictable process. Moreover, we have
 $$ S(t) = f_0(t) + \sum_{n\in I} Y_n(t)$$
 where $dY_n(t) = (\mu_n(t)-\gamma_nY_n(t))dt + \lambda_n\sigma(t)dB_n(t)$ for any $t\in\mathbb R_+$.
\end{cor}
\begin{proof}
 We start to calculate the covariance operator $Q$ of $W$. Let $V\subseteq H_w$ be the vector space generated by $g_1,\dots,g_d$. Then for any $h,g\in H_w$ where $h$ orthogonal to $V$ we have $\<W(1),g\> = 0$ and hence
 $$ \< \Q h,g\> = \E(\<W(1),h\>\<W(1),g\>) = 0. $$
 Let $(g^*_1,\dots,g^*_d)$ be the biorthogonal system for $(g_1,\dots,g_d)$ in $V$. For $n,k\in I$ we have
 $$\<W(1),g_n^*\> = B_n(1)$$ and hence $ \< \Q g^*_n,g^*_k\> = 1_{\{n=k\}}.$ Consequently, $\Q g = \sum_{n=1}^d \<g,g_n\>g_n$.

 Now we show that $\beta$ takes values in the vector space $V\subseteq H_w$ generated by $g_1,\dots,g_d$. The fundamental theorem of asset pricing by Delbaen and
Schachermayer~\cite[Main theorem 1.1]{delbaen.schachermayer.98} yields an equivalent probability measure $Q$ such that $(F(t,T))_{0\leq t\leq T}$ is a $\sigma$-martingale under $Q$ for any $T>0$. Since the $P$-dynamics are given by
 $$ F(t,T) = \U_Tf_0 + \int_0^t \U_{T-t}\beta(s)ds + \int_0^t \U_{T-s}\sigma(s)dW(s) $$
 we know, by using Jacod and Shiryaev~\cite[Theorem III.3.24]{js.87}, that $W(t) = W^Q(t) + \int_0^t\beta^Q(s)ds$  where $W^Q$ is a $Q$-Brownian motion and $\beta^Q:\mathbb R_+\times\Omega\rightarrow V$ is some predictable process. Hence the $Q$-dynamics are given by
 $$ F(t,T) = \U_Tf_0 + \int_0^t \beta_Q(s)ds + \int_0^t \sigma(s)\U_{T-s}dW^Q(s)$$
 where $\beta_Q(t):=\beta(t)+\sigma(t)\U_{T-t}\beta(t)$, $t\geq 0$ is some predictable process. However, since $F(t,T)$, $t\in[0,T]$ is a $\sigma$-martingale under $Q$ we have $\beta_Q=0$. Thus, $\U_{T-t}\beta(t) = -\sigma(t)\U_{T-t}\beta(t) \in V$ for any $0\leq t\leq T$.
\end{proof}

Several remarks are in place in connection to the Corollary above. First of all,
we recall that typical spot price models in commodity markets are given as sums
of Ornstein-Uhlenbeck processes. Lucia and Schwartz ~\cite{lucia.schwart.02}
propose, among other models, a Brownian driven Ornstein-Uhlenbeck dynamics for the spot
price of power in the Nordic electricity market NordPool. They suggest a one or two-factor
model, where in the latter case the Ornstein-Uhlenbeck process degenerates to drifted
Brownian motion.  Benth, Kallsen and Meyer-Brandis~\cite{benth.et.al.05}
propose a general class of models for power prices defined as a sum of Ornstein-Uhlenbeck
processes with different speeds of mean reversion, driven by L\'evy processes.
Corollary~\ref{cor:OU-repr} tells us that our specific choice of the volatility 
structure $\Psi$ in the forward curve dynamics implies such a spot dynamics. 

The representation of the forward curve in terms of a finite sum of Ornstein-Uhlenbeck
processes can also be seen as a special case of the theory of finite dimensional 
realisations of forward rate models in fixed-income theory. In the paper of
Bj\"ork and Gombani~\cite{BG}, it is shown that the solution of the stochastic partial 
differential equation for forward rates can be represented as a linear combination of 
Ornstein-Uhlenbeck processes for special affine-like volatility term structures. Their result
matches very well with Corollary~\ref{cor:OU-repr} above. The theory of Bj\"ork and Gombani~\cite{BG} has later been significantly generalized by Filipovi\'c and Teichmann~\cite{FT1,FT2}, Ekeland and Taflin~\cite{ET}, and more recently by Tappe~\cite{Tappe.10}. The theory of Tappe~\cite{Tappe.10} has been modified and slightly extended to forward curve models in commodity markets by Benth and Lempa~\cite{BL}, where the authors apply it to optimal portfolio management.   

A simulation scheme for the forward price dynamics based on theory for
numerical solution of hyperbolic stochastic differential equations have been desgined and
analysed in Barth and Benth~\cite{BB}. The method applies a decomposition of the 
covariance operator along with a finite element method.


\section{Operators on $H_{w}$}


In this Section we provide an in-depth analysis of various classes 
of operators on the space $H_w$ which are relevant in connection with forward 
price modeling. In the dynamics of the forward price $f(t)$ in \eqref{G:f Volterra representation}, we have
operators present in the volatility, $\Psi$ and the the covariance 
$\Q$ of the noise $L$. Positive trace class operators play an important role for square integrable L\'evy processes, 
cf.\ Peszat and Zabczyk~\cite[Sections 4.4, 4.6]{peszat.zabczyk.07}. These operators are squares of symmetric Hilbert-Schmidt operators, and we provide a complete characterisation of these. It turns out that Hilbert-Schmidt operators on $H_w$ are almost integral operators. Moreover, we analyse the particular cases of integral and 
multiplication operators, which are highly relevant for 
concrete modelling purposes in commodities and energy.  

Throughout this Section, the operator $\mathcal W$ mapping $H_w$ into $L^2(\mathbb R_+)$ defined by
\begin{equation}
\label{r:quadratintegrierbarkeit}
\mathcal W f=\sqrt{w}f'\,,
\end{equation}
will become useful. It holds that 
\begin{equation}
\label{r:quadratintegrierbarkeit:isometrie}
(\delta_0,\mathcal W):H_w\rightarrow \mathbb R\times L^2(\mathbb R_+),f\mapsto(f(0),\sqrt{w} f')
\end{equation}
 is an isometric isomorphism of the Hilbert spaces $H_w$ and $\mathbb R\times L^2(\mathbb R_+)$.

\subsection{Integral operators}\label{subsect:int_op}

A useful class of operators are integral operators. Naturally, an integral operator
on  $H_w$  is 
defined as 
\begin{equation}
\mathcal If(x)=\int_0^{\infty}r(x,y)f(y)\,dy
\end{equation}
for elements $f\in H_w$. Obviously, to have $\mathcal If\in H_w$ depends on 
the integral kernel function $r:\mathbb R_+^2\mapsto\mathbb R$. However,
using the representation
$$
f(y)=f(0)+\int_0^yf'(z)\,dz\,,
$$
we find
\begin{align*}
\mathcal If(x)&=\int_0^{\infty}r(x,y)\,dy\,f(0)+\int_0^{\infty}r(x,y)\int_0^yf'(z)\,dz\,dy \\
&=\int_0^{\infty}r(x,y)\,dy\,f(0)+\int_0^{\infty}\int_z^{\infty}r(x,y)\,dy f'(z)\,dz\,.
\end{align*}
The first term in this representation is a multiplication of a function given by the kernel $r$
with the evaluation of $f$ at zero. The second term is again an integral operator, but now of 
the form
$$
\mathcal{T} f(x)=\int_0^{\infty}q(x,y)f'(y)\,dy  
$$
for some kernel function $q:\mathbb R_+^2\mapsto\mathbb R$. We study these in the sequel,
and consider mutliplication operators in Subsection~\ref{a:muliplikatoren}.

\begin{defn}
 Let $q:\mathbb R_+^2\rightarrow\mathbb R$ be Borel measurable. The {\em integral operator on $H_w$ with kernel $q$} is 
 defined on its domain $\dom(\T)$ for those functions $f\in H_w$ which satisfy
\begin{enumerate}
 \item $\int_0^\infty \vert q(x,y)f'(y)\vert dy<\infty$ for any $x\in\mathbb R_+$ and
 \item $\left(x\mapsto \int_0^\infty q(x,y)f'(y)dy\right) \in H_w$
\end{enumerate}
and it is given by
   $$ \T f(x) := \int_0^\infty q(x,y)f'(y)dy,\quad f\in\dom(\T), x\in\mathbb R_+.$$
\end{defn}
For example,
one can realize the noise field $L$ on $H_w$ and let its covariance
operator $\Q$ be represented as an integral operator. This means 
that we define $\Q$ via a kernel function $q$ as in the definition
above. This kernel function gives a functional description
of the correlation between two neighbouring locations along the curve $L(t)$.
Also, letting $\Psi$ in the dynamics of $f$ in \eqref{G:f dynamik} be an integral operator, enables us to mix the noise $L$ with a kernel function $q$, which in a sense is scaling the various noise sources for 
different locations to make up the noise in one. This is a natural 
generalisation of the case of finite-dimensional noise, where one is
typically having a volatility which is a linear combination of the various 
components in the noise vector. 

In order to understand integral operators we proceed in the following way.

\begin{rem}\label{r:ueberall definiert}
 If $\T$ is an integral operator with kernel $q$, then it is eveywhere defined, i.e.\ its domain equals $H_w$ if and only if
\begin{itemize}
 \item[(1')] $q(x,\cdot)/\sqrt{w}\in L^2(\mathbb R_+)$ for any $x\in\mathbb R_+$ and
 \item[(2')] $\left(x\mapsto \int_0^\infty q(x,y)f'(y)dy\right) \in H_w$ for any $f\in H_w$.
\end{itemize}
\end{rem}

The condition (1') ensures that $f_q(x):=\int_0^{\infty} q(x,y)f'(y)dy$, $x\in\mathbb R_+$ defines a measurable function. While condition (2') ensures that $f_q$ is actually an element of $H_w$.

Sometimes an integral operator $\T$ with kernel $q$ can be extended to those functions $f\in H_w$ where the integral $\int_0^\infty q(x,y)f'(y)dy$ makes sense for any $x\geq 0$. However, we cannot expect that the resulting functions are members of $H_w$ or even continuous but as we show next they are measurable. This technical extension is needed later.
\begin{lem}\label{l:messbarkeit}
 Let $\T$ be an integral operator with kernel $q$ and $f\in H_w$ such that 
$$\int_0^\infty \vert q(x,y)f'(y)\vert dy<\infty$$ for any $x\in\mathbb R_+$. Define
$$ f_q(x):= \int_0^\infty q(x,y)f'(y)dy$$
for any $x\in\mathbb R_+$. Then $f_q$ is a measurable function.
\end{lem}
\begin{proof}
 Let $r_n:\mathbb R_+^2\rightarrow\mathbb R$ such that $r_n$ is elementary $r_n\rightarrow q$ pointwise and $\vert r_n(x,y)\vert\leq 2\vert q(x,y)\vert$. Let $g_n:\mathbb R_+\rightarrow\mathbb R$ such that $g_n$ is elementary, $g_n\rightarrow f'$ pointwise Lesbesgue-a.e.\ and $\vert g_n(x)\vert\leq 2\vert f'(x)\vert$. Then $r_ng_n\rightarrow qf'$ and $\vert r_ng_n\vert \leq 4\vert qf'\vert$. Thus the dominated convergence theorem yields
 $$ h(x) := \int_0^\infty q(x,y)f'(y) dy = \lim_{n\rightarrow\infty} \int_0^\infty r_n(x,y)g_n(y)dy.$$
 However, $h_n:\mathbb R_+\rightarrow\mathbb R,x\mapsto\int_0^\infty r_n(x,y)g_n(y)dy$ is elementary and hence measurable. Thus $h$ is the pointwise limit of measurable functions and hence measurable.
\end{proof}

Recall, that an operator $\T$ is defined to be closed if for any sequence $(f_n)_{n\in\mathbb N}$ in the domain of $\T$ such that $f_n$ converges to some $f\in H_w$ and $\T f_n$ converges to some $g\in H_w$ we have $f$ is in the domain of $\T$ and $\T f=g$. Typically, on $L^2(\mathbb R_+)$ integral operators, which are exactly defined for those functions where the integral expression yields an $L^2$-function, are {\em not closed!}, cf.\ Grafakos~\cite[Theorem I.1]{grafakos.08}. However, an integral operator $\T$ on $H_w$ satisfying condition (1') is closed as we now show.
\begin{lem}\label{l:abgeschlossenheit von T}
 Let $\T$ be an integral operator with kernel $q$. Assume that (1') in Remark \ref{r:ueberall definiert} is satisfied, i.e.\ $q(x,\cdot)/\sqrt{w}\in L^2(\mathbb R_+)$ for any $x\in\mathbb R_+$. Then $\T$ is a closed linear operator.
\end{lem}
\begin{proof}
 Linearity of $\T$ is obvious. Define $b(x,y):=q(x,y)/\sqrt{w(y)}$ for any $x,y\in\mathbb R_+$. Then $b(x,\cdot)\in L^2(\mathbb R_+)$ for any $x\in\mathbb R_+$ by assumption. Let $(f_n)_{n\in\mathbb N}$ be a sequence of elements in the domain of $\T$ which converges in $H_w$ to some $f$ such that 
$\T f_n$ converges to some $g$ in $H_w$. Recall the 
operator $\mathcal W$ defined in \eqref{r:quadratintegrierbarkeit}, which is a continuous
linear operator.
Thus $\mathcal W f_n\rightarrow \mathcal W f$ in $L^2$. Let $x\in\mathbb R_+$. We have
 $$q(x,\cdot)f'_n = b(x,\cdot)\mathcal W f_n \rightarrow b(x,\cdot)\mathcal W f = q(x,\cdot)f'$$
 where the convergence is in $L^1(\mathbb R_+)$. Hence
 \begin{eqnarray*}
    \T f_n(x) &=& \int_0^\infty q(x,y)f'_n(y) dy \\
            &\rightarrow& \int_0^\infty q(x,y)f'(y)dy \\
            &=:& h(x).
 \end{eqnarray*}
 Lemma \ref{l:stetigkeit der Punktauswertung} yields that 
$\T f_n(x)\rightarrow g(x)$. Thus $h = g$. Consequently, 
$h\in H_w$ which implies that $f\in\dom(\T)$ and $g=\T f$.
\end{proof} 

The following corollary shows that an integral operator is continuous if and only if it is everywhere defined. This is a huge improvement compared to the statement in Remark \ref{r:ueberall definiert}. However, boundedness conditions for integral operators which are not everywhere defined are much more involved.
\begin{cor}\label{k:stetigkeit}
 Let $\T$ be an integral operator with kernel $q$. 
$\T$ is continuous if and only if the following two statements 
are satisfied.
\begin{enumerate}
 \item $q(x,\cdot)/\sqrt{w}\in L^2(\mathbb R_+)$ for any $x\in\mathbb R_+$.
 \item $\left(x\mapsto \int_0^\infty q(x,y)f'(y)dy\right) \in H_w$ for any $f\in H_w$.
\end{enumerate}
\end{cor}
\begin{proof}
 Assume that the two conditions are satisfied. Then $\T$ 
is everywhere defined by definition. Lemma \ref{l:abgeschlossenheit von T} 
yields that $\T$ is a closed operator. The closed graph theorem 
implies that $\T$ is a continuous linear operator.

 Now assume that $\T$ is a continuous linear operator. 
Let $f\in H_w$ and $x\in\mathbb R_+$. Then $f$ is in the domain of 
$\T$ and consequently it satisfies condition (2) by definition 
and
$$
\int_0^\infty \vert q(x,y)/w(y) (wf')(y))\vert dy = \int_0^\infty \vert q(x,y)f'(y)\vert dy<\infty.
$$ 
Since this is true for any $f\in H_w$ we conclude that 
$y\mapsto q(x,y)/w(y) \in L^2(\mathbb R_+)$. This proves the Corollary.
\end{proof}
Let us consider an example which constitute an important case for power markets. 
In our definition of the forward curve dynamics $f(t)$, see \eqref{G:f dynamik}, 
we model the forward price of contracts with a {\it fixed} time to delivery
$x\in\mathbb R_+$. 
In the power market one is trading in contracts delivering over a
given {\it period} of time, and these contracts can be viewed as holding a
portfolio of forwards delivering at each fixed time instant in the 
delivery period (see Benth and Koekebakker~\cite{BK}). 
If we denote by $F^{\tau}(t,x)$ the forward price
at time $t$ of a contract with time to delivery $x\in\mathbb R_+$ and
{\it length} of delivery $\tau\in\mathbb{R}_+$, we find that
\begin{equation}
F^{\tau}(t,x)=\frac{1}{\tau}\int_x^{x+\tau}f(t,y)\,dy\,.
\end{equation}
The forward price $F^{\tau}(t,x)$ is the average of the function 
$y\mapsto f(t,y)$ over the 
delivery period $[x,x+\tau]$ by market convention. From the representation 
$$
f(t,y)=f(t,x)+\int_x^y\frac{\partial f}{\partial x}(t,z)\,dz
$$ 
for $y\geq x$, we find from Fubini's theorem
\begin{align*}
F^{\tau}(t,x)&=\frac1{\tau}\int_x^{x+\tau}\left(f(t,x)+\int_x^y
\frac{\partial f}{\partial x}(t,z)\,dz\right)\,dy \\
&=f(t,x)+\frac1{\tau}\int_x^{x+\tau}\int_z^{x+\tau}\,dy\frac{\partial f}{\partial x}(t,z)\,dz \\
&=f(t,x)+\int_0^{\infty}q^{\tau}(x,y)\frac{\partial f}{\partial x}(t,y)\,dy
\end{align*}
with
$$
q^{\tau}(x,y)=\frac1{\tau}(x+\tau-y)\mathbf{1}_{[x,x+\tau]}(y)\,.
$$
Denote by $\T^{\tau}$ the integral operator with kernel
$q^{\tau}$. We observe that the function $y\mapsto q^{\tau}(x,y)$ has support 
on $[x,x+\tau]$ for each given $x\in \mathbb R_+$, and hence Condition (1) 
in Corollary~\ref{k:stetigkeit} is trivially satisfied  
since $w$ is a continuously increasing function with $w(0)=1$.
Moreover, since
$$
(\T^{\tau}f)'(x)=\frac1{\tau}\frac{d}{dx}\int_x^{x+\tau}(x+\tau-y)
f'(y)\,dy=\frac{f(x+\tau)-f(x)}{\tau} - f'(x)
$$
it holds that
\begin{align*}
 \int_0^\infty w(x) ((\T^\tau f)'(x)+f'(x))^2dx 
                   &\leq \frac{1}{\tau}\int_0^\infty w(x) \int_x^{x+\tau} (f'(y))^2\frac{dy}{\tau}dx \\
                   &\leq \frac{1}{\tau}\int_0^\infty \int_x^{x+\tau} w(y)(f'(y))^2\frac{dy}{\tau}dx \\
                   &\leq \frac{1}{\tau}\int_0^\infty  w(y)(f'(y))^2 dy < \infty, 
\end{align*}
for every $f\in H_w$, where we used respectively Jensen's inequality, that $w$ is increasing 
and Tonelli's theorem. Hence,
by the triangle inequality we find,
\begin{align*}
 \Vert \T^{\tau}f \Vert_w &\leq \Vert \T^{\tau} f + f\Vert_w + \Vert f \Vert_w <\infty.
\end{align*}
This shows that Condition (2) of 
Corollary~\ref{k:stetigkeit} is satisfied, which implies that 
$\T^{\tau}$ is a continuous integral operator on $H_w$. Moreover,
it follows that $F^{\tau}(t)\in H_w$ for any forward curve $f(t)\in H_w$. 
Hence, we conclude that the power forward price dynamics for given length 
of delivery $\tau$ can be realized in $H_w$ as an application of a 
continuous integral operator along with the identity 
map on forward curves realized in $H_w$:
$$
F^{\tau}(t)=(\text{Id}+\T^{\tau})(f(t))
$$ 
We note here that $\tau$, the length of delivery, is viewed as a
parameter. One may in fact introduce this as a second variable, and 
introduce function spaces on $\mathbb R_+^2$ in a similar fashion as 
$H_w$. We leave this for forwards studies.

We continue with the general analysis of integral operators on $H_w$. 
In order to show that an integral operator defined by a function $q$ is bounded it is useful to consider Schur's lemma (see Grafakos~\cite[Appendix I.1]{grafakos.08}).
Remark that this Lemma is sometimes known as {\it Schur's test}. However, we first have to identify integral operators on $H_w$ with integral operators on $L^2(\mathbb R_+)$.
\begin{thm}\label{S:stetige integraloperatoren}
 Let $\T$ be an integral operator on $H_w$ with kernel $q$. 
Assume that $q$ has an absolute derivative with respect to its first 
component, define $$b(x,y):=\partial_1q(x,y)\sqrt{w(x)/w(y)}$$ and assume that 
\begin{enumerate}
 \item $q(0,\cdot)/\sqrt{w}\in L^2(\mathbb R_+)$,
 \item $\sup_{x\in\mathbb R}\int_0^\infty \vert b(x,y)\vert dy=:A<\infty$ and
 \item $\sup_{y\in\mathbb R}\int_0^\infty \vert b(x,y)\vert dx=:B<\infty$.
\end{enumerate}
Then $\T$ is densely defined and bounded by 
$$
c := \left(\int_0^\infty\frac{q^2(0,y)}{w(y)}dy+ AB\right)^{1/2}\,.
$$ 
Consequently, $\T$ can be uniquely extended to a continuous 
linear operator. Moreover, $b$ is locally integrable in its second 
variable and hence we can define the function
$$
q^*(y,x) := \int_0^y\sqrt{\frac{w(x)}{w(z)}}b(x,z)dz\,.
$$
The dual of $\T$ is a continuous linear operator, bounded by $c$ and 
$$
\T^*g(y) = \left(\int_0^y\frac{q(0,z)}{w(z)}dz\right)g(0) 
+ \int_0^\infty q^*(y,x)g'(x) dx\,,
$$
for $g\in H_w$, $y\in\mathbb R_+$ such that 
$wg'\in L^1(\mathbb R_+)\cap L^\infty(\mathbb R_+)$.

Moreover, if additionally 
$$
(1')\quad q(x,\cdot)/\sqrt{w}\in L^2(\mathbb R_+)\,,
$$
for any $x\in\mathbb R_+$, then $\T$ is a continuous linear operator. 
If $q^*(y,\cdot)/\sqrt{w}\in L^2(\mathbb R_+)$ for any $y\in\mathbb R_+$, 
then the representation for $\T^*$ extends to all functions in $H_w$.
\end{thm}
\begin{proof}
  Schur's lemma, cf.\ Grafakos~\cite[Appendix I.1]{grafakos.08}, yields that
 $$
\widetilde{\mathcal R}f(x) := \int_0^\infty b(x,y)f(y)dy,\quad f\in 
D(\widetilde{\mathcal R}),x\in\mathbb R_+$$
 is an operator on $L^2(\mathbb R_+)$ with bound $(AB)^ {1/2}$, 
whose domain includes $L^1(\mathbb R_+)\cap L^\infty(\mathbb R_+)$. 
Since 
$$
\partial_1q(x,y) = \sqrt{w(y)/w(x)}b(x,y)\,,
$$ we conclude from Condition (2) and from local boundedness of $w$ that 
$\partial_1q$ is locally integrable.

Let 
$$
\widetilde{\mathcal W}:L^2(\mathbb R_+)\rightarrow H_w,g\mapsto \int_0^x (w(z))^{-1/2}g(z)dz\,.
$$
Then $\widetilde{\mathcal W}$ is an isometry and 
$\mathcal W\widetilde{\mathcal W}$ is the identity on 
$L^2(\mathbb R_+)$ where $\mathcal W$ is defined in 
\eqref{r:quadratintegrierbarkeit}. Thus
 $\mathcal R := \widetilde{\mathcal W} \widetilde{\mathcal R}\mathcal W$ 
is a densely defined operator with bound $(AB)^{1/2}$. Let $f\in H_w$ 
such that $f'$ is continuous and has compact support $C$. 
Then $\mathcal W f \in \dom(\widetilde{\mathcal R})$ and hence 
$f\in\dom(\mathcal R)$. Thus
 \begin{eqnarray*}
  \mathcal Rf(x) &=& \int_0^x (w(z))^{-1/2} \int_0^\infty b(z,y)\sqrt{w(y)}1_C(y)f'(y)dydz  \\
                    &=& \int_0^x \int_C \partial_1q(z,y)f'(y)dydz \\
                    &=& \int_C\int_0^x\partial_1q(z,y)dzf'(y)dy \\
                    &=& \int_C (q(x,y)-q(0,y)) f'(y)dy
 \end{eqnarray*}
where we used Fubini's theorem for the second last equation. 
Condition (1) implies that 
$$
\int_0^\infty \vert q(0,y)f'(y)\vert dy<\infty\,,
$$ 
and hence we have $\mathcal Rf(x) + \T f(0) = \T f(x)$, 
$x\in\mathbb R_+$. The vector space 
$$
V:=\{ f\in H_w: \mathcal W f\in L^1(\mathbb R_+)\cap L^\infty(\mathbb R_+)\}\,,
$$
is dense in $H_w$ because \eqref{r:quadratintegrierbarkeit:isometrie} 
states that $(\delta_0,\mathcal W)$ is an isometry, 
$$
(\delta_0,\mathcal W)(V) = \mathbb R\times \left(L^1(\mathbb R_+)\cap L^\infty(\mathbb R_+)\right)
$$
and $L^1(\mathbb R_+)\cap L^\infty(\mathbb R_+)$ is dense in 
$L^2(\mathbb R_+)$. Since $V$ is contained in the domain of $\T$, 
we have that $\T$ is densely defined. 
Now let $f\in\dom(\T)$ with $\Vert f\Vert\leq 1$. 
Then the Pythagorean theorem, the Cauchy-Schwarz inequality and the 
bound for $\mathcal R$ imply
 $$ \Vert\T f\Vert^2 = \Vert (x\mapsto \T f(0))\Vert^2 + 
\Vert\mathcal Rf\Vert^ 2\leq \int_0^\infty \frac{q^2(0,y)}{w(y)}dy+AB.
$$

Observe, that $q^*$ also satisfies conditions (1), (2) and (3). Hence 
the integral operator $\mathcal R^*$ with kernel $q^*$ is densily 
defined and bounded by $(AB)^{1/2}$. Moreover, analogue computations 
to those above show that
 $$ 
\mathcal W\mathcal R^*\widetilde{\mathcal W} g(x) = \int_0^\infty b(y,x)g(y) dy\,,
$$
for bounded compactly supported functions $g$, and hence 
$\mathcal W\mathcal R^*\widetilde{\mathcal W}$ is the dual of 
$\widetilde{\mathcal R}$. Consequently, $\mathcal R^*$ is the dual 
operator of $\mathcal R$. Now it is easy to deduce the dual operator of 
$\T$.

Lemma \ref{l:abgeschlossenheit von T} states that the additional assumption 
implies that $\T$ is closed. However, a closed, bounded and 
densely defined operator is everywhere defined and continuous.
\end{proof}
\begin{rem}
 In the special case that $q(x,y)=q(0,y)$ we see that the dual 
of $\T$ is an operator which takes the initial value of the 
function and multiplies it with
$$
h(y):=\int_0^y\frac{q(0,z)}{w(z)}dz\,.
$$
Hence, one might write $\T^* =\mathcal M_h\delta_0$, where 
$\mathcal M_h$ denotes the multiplication operator by the function $h$. We will return
to a more detailed study of general multiplication operators in Subsection~\ref{a:muliplikatoren}.
\end{rem}

A prominent class of integral operators are convolution type operators, 
i.e.\ operators of the form
 $$ 
\T f(x) = \int_0^\infty k(x-y)f'(y)dy\,,
$$
or, in other words, an integral operator with kernel $q(x,y):=k(x-y)$. 
The next Corollary restates the conditions of 
Theorem~\ref{S:stetige integraloperatoren} for such convolution type 
operators. However, for simplicity, we will only do this in the case that 
the weight function $w$ in the space $H_w$ is an exponential function.
\begin{cor}
  Let $k:\mathbb R\rightarrow\mathbb R_+$ be a measurable function 
and $\T$ be an integral operator on $H_w$ with kernel 
$q:\mathbb R_+^2\rightarrow \mathbb R,(x,y)\mapsto k(x-y)$. 
Assume that 
\begin{enumerate}
 \item $w(x) = e^{\alpha x}$ for some $\alpha>0$ and any $x\in\mathbb R_+$,
 \item $\int_0^\infty \vert k(x-y)\vert^2e^{-\alpha y}dy <\infty$ for any $x\in\mathbb R_+$,
 \item $k$ has an absolutely continuous derivative and
 \item $\int_{-\infty}^\infty\vert k'(s)\vert e^{\alpha s/2}ds <\infty$.
\end{enumerate}
Then $\T$ is a continuous linear operator.
\end{cor}
\begin{proof}
 Define $$b(x,y):=\partial_1q(x,y)\sqrt{w(x)/w(y)},\quad x,y\in\mathbb R_+.$$
 Then $b(x,y) = k'(x-y)e^{1/2\alpha(x-y)}$ for any $x,y\in\mathbb R_+$. Let $x\in\mathbb R_+$, then
 $$\int_0^\infty \vert b(x,y)\vert dy = \int_{-\infty}^x\vert k'(s)\vert e^{1/2\alpha s}ds$$
 and hence
 $$ \sup_{x\in\mathbb R_+}\int_0^\infty \vert b(x,y)\vert dy = \int_{-\infty}^\infty\vert k'(s)\vert e^{1/2\alpha s}ds.$$
 Along the same lines we get
 $$ \sup_{y\in\mathbb R_+}\int_0^\infty \vert b(x,y)\vert dx = \int_{-\infty}^\infty\vert k'(s)\vert e^{1/2\alpha s}ds.$$
 We have
 \begin{align*}
   \int_0^\infty \left\vert q(x,y)/\sqrt{w(y)}\right\vert^2 dy &\ = \int_0^\infty \vert k(x-y)\vert^2e^{-\alpha y}dy
 \end{align*}
 for any $x\in\mathbb R_+$. The claim follows from Theorem \ref{S:stetige integraloperatoren}.
\end{proof}
Let us end this Subsection on integral operators by coming back to the spot dynamics
as presented in Thm.~\ref{s:dynamik}. Recall that the spot volatility is given by
$$
\sigma^2(t,s)=\<\Psi(s)\mathcal Q \Psi^*(s)h_{t-s},h_{t-s}\>\,.
$$
Here, $\mathcal{Q}$ is the covariance operator of the driving noise of the 
forward dynamics $f(t)$ and $\Psi$ is the volatility operator in its 
dynamics. If we assume now that $\Psi(s)=\widetilde{\sigma}(s)\mathcal T$, for an 
integral operator on $H_w$
$$
\mathcal{T}g(x)=\int_0^{\infty}\psi(x,y)g'(y)\,dy\,.
$$
Observe that $\widetilde{\sigma}$ is a predictable $\mathbb R$-valued stochastic
process. 
We also assume $\mathcal{Q}$ to be an integral operator on $H_w$, cf.\ Theorem \ref{S:Darstellung von HS operatoren}
$$
\mathcal{Q}g(x)=\int_0^{\infty}q(x,y)g'(y)\,dy\,.
$$
Under the assumptions above we can compute $\sigma^2 (t,s)$ 
and find that
$$
\sigma^2(t,s)=\widetilde{\sigma}^2(s)\int_0^{\infty}\int_0^{\infty}\frac1{w(u)}\psi(t-s,z) q_1(z,u)\psi(t-s,u)\,du\,dz
$$ 
where $q_1$ is the partial derivative of $q$ with respect to its first argument. Thus, 
$\sigma^2(t,s)$ will be of the form
$$
\sigma(t,s)=\widetilde{\sigma}(s)\gamma(t-s)\,,
$$
which gives an LSS process (see Barndorff-Nielsen et al.~\cite{BNBV-spot}). We see that 
the {\it deterministic} kernel function $\gamma(t-s)$ will be an integral in terms
of the kernels from both the volatility operator $\Psi$ and the covariance operator
$\mathcal{Q}$. A simplistic special case could be to assume that $\psi(x,y)=\xi(x)\theta(y)$,
in which case we find
$$
\sigma(t,s)=\widetilde{\sigma}(s)c\xi(t-s)
$$  
for a constant $c$. In this simple example we can recover interesting LSS processes for
the spot price dynamics by choosing $\xi$ appropriately. For example, we
can recover so-called continuous-time autoregressive moving average (CARMA) processes by
letting
$$
\xi(x)=\mathbf{b}^*\exp(A x)\mathbf{e}_p\,.
$$
Here, $\mathbf{e}_p$ is the $p$'th canonical basis vector in $\mathbb R^p, p\in\mathbb N$,
$\mathbf b\in \mathbb R^p$ is a vector $\mathbf{b}^*=(b_0,b_1,\ldots,b_{q}=1,0,..,0)$
for $q<p$, and $A$ is a $p\times p$-matrix of the form
$$
A=\left[\begin{array}{cc} 0 & I  \\  -\alpha_p & -\alpha_{p-1} \cdots -\alpha_1    \end{array}\right],
$$
with $\alpha_i>0$, $i=1,\ldots,p$ and $I$ the $p-1\times p-1$ identity matrix. 
We say that we have a volatility modulated CARMA($p,q$) dynamics for the spot for this
special choice. Note that 
we are free to choose the covariance kernel $q$ and the $\theta$, which means that we can 
have many different forward curve models resulting in the same CARMA spot model. 
CARMA-processes have been applied to temperature models (see Benth and 
\v{S}altyt\.{e} Benth.~\cite{BSB-book} and H\"ardle and Lopez-Cabrera \cite{HLC}),
spot power prices (see Garcia, Kl\"uppelberg and M\"uller~\cite{GKM}) and the prices
of oil (see Paschke and Prokopczuk~\cite{PP}). For a detailed analysis of CARMA processes
we refer the reader to Brockwell \cite{Brock}.

\subsection{Hilbert-Schmidt operators}
In this Section we analyse the Hilbert-Schmidt operators on $H_w$. 
The Hilbert-Schmidt operators on $L^2$ are exactly the integral operators where the kernel is an $L^2$ function and the Hilbert-Schmidt norm is their $L^2$-norm. This can be exploited to identify the Hilbert-Schmidt operators on $H_w$ in terms of integral operators, as $H_w$ is almost isometric to $L^2$ .
\begin{lem}\label{L:HS fuer L^2}
 Let $\T$ be a Hilbert-Schmidt operator on $L^2((\mathbb R_+,\mathcal B,\lambda),\mathbb R)$. Then there is $b\in L^2((\mathbb R^2_+,\mathcal B_2,\lambda_2),\mathbb R)$ such that
\begin{eqnarray}\label{e:HS operatoren auf L^2}
\T f(x) = \int_0^\infty b(x,y)f(y) dy 
\end{eqnarray}
 for any $f\in L^2((\mathbb R_+,\mathcal B,\lambda),\mathbb R)$ where $\lambda_2$ denotes the two-dimensional Lebesgue measure. The Hilbert-Schmidt norm of 
$\T$ is the $L^2$-norm of $b$.

Moreover, if $b$ is an element of $L^2((\mathbb R^2_+,\mathcal B_2,\lambda_2),\mathbb R)$, then
 $$\T f(x) := \int_0^\infty b(x,y)f(y) dy$$
defines a Hilbert-Schmidt operator.
\end{lem}
\begin{proof}
 This is immediate from the isomorphisms from the Hilbert-Schmidt operators to
$$H:=L^2((\mathbb R_+,\mathcal B,\lambda),\mathbb R)\otimes L^2((\mathbb R_+,\mathcal B,\lambda),\mathbb R)$$
and from $H$ to $L^2((\mathbb R^2_+,\mathcal B_2,\lambda_2),\mathbb R)$.
\end{proof}
The next Lemma shows that the Hilbert-Schmidt operators on a subspace of $H_w$ with codimension $1$ are integral operators: 
\begin{lem}\label{K:HS auf G}
 Let $\mathcal C$ be a closed linear operator on $H_w^0:=\{f\in H_w:f(0)=0\}$. Then the following statements are equivalent.
\begin{enumerate}
 \item $\mathcal C$ is a Hilbert-Schmidt operator.
 \item $\mathcal C$ is an integral operator and there is $b\in L^2(\mathbb R_+^2)$ such that the kernel of $\mathcal C$ is given by
  \begin{eqnarray}\label{e:q und b}
   q(x,y) := \int_0^x\sqrt{w(y)/w(z)}b(z,y)dz,\quad x,y\in\mathbb R_+.   
  \end{eqnarray}
\end{enumerate}
If the second statement holds, then the Hilbert-Schmidt norm coincides with the $L^2(\mathbb R_+^2)$-norm of $b$.
\end{lem}
\begin{proof}
  Assume that $\mathcal C$ is a Hilbert-Schmidt operator. Then 
$\T:=\mathcal W\mathcal C\mathcal W^{-1}$ is a Hilbert-Schmidt operator where 
$\mathcal W$ is the ismorphism from \eqref{r:quadratintegrierbarkeit} restricted to $H_w^0$ and $\T$ and $\mathcal C$ have the same Hilbert-Schmidt norm. Lemma \ref{L:HS fuer L^2} yields that there is a square integrable function $b$ such that 
$\T f(x) = \int_0^\infty b(x,y) f(y) dy$ and the Hilbert-Schmidt norm 
coincides with the $L^2$-norm of $b$. Since
$$
\mathcal W^{-1}g(x)=\int_0^x\frac{g(y)}{\sqrt{w(y)}}\,dy\,,
$$
a short computation yields the representation of 
$\mathcal C$. Indeed, for $f\in H_w^0$ we have
 \begin{eqnarray*}
   \mathcal Cf(x) &=& \mathcal W^{-1}\T\mathcal W f(x) \\
         &=& \int_0^x\frac{1}{\sqrt{w(y)}}\T\mathcal W f(y) dy \\
         &=& \int_0^x\frac{1}{\sqrt{w(y)}} \int_0^\infty b(y,z)\mathcal W f(z)dz dy \\
         &=& \int_0^x\frac{1}{\sqrt{w(y)}} \int_0^\infty b(y,z) \sqrt{w(z)} f'(z)dz dy \\
         &=& \int_0^\infty q(x,z)f'(z) dz
 \end{eqnarray*}
 where we used Fubini's theorem and $q$ is defined by Equation \eqref{e:q und b}.

 Now assume that $\mathcal C$ is an integral operator given by a function as in condition (2). Then $\T$ defined by Equation \eqref{e:HS operatoren auf L^2} is a Hilbert Schmidt-operator. The same equations as before yield that 
$\mathcal C=\mathcal W^{-1}\T\mathcal W$.
\end{proof}
The next Theorem gives a complete characterisation of Hilbert-Schmidt operators on
$H_w$. The proof makes use of the above Lemma. 
\begin{thm}\label{S:Darstellung von HS operatoren}
 Let $\mathcal C$ be a Hilbert-Schmidt operator on $H_w$. Then there is a square-integrable function $b:\mathbb R_+^2\rightarrow\mathbb R$, $g,h\in H_w$ with $g(0)=0=h(0)$ and $c\in\mathbb R$ such that
 $$\mathcal Cf(x) = cf(0) + \<g,f\> + f(0)h(x) + \int_0^\infty q(x,z)f'(z) dz$$
 where $q(x,z) :=  \int_0^x\sqrt{w(z)/w(y)}b(y,z)dy.$ Moreover, the Hilbert-Schmidt norm of $\mathcal C$ is given by
 $$ \Vert\mathcal C\Vert^2_{HS} = c^2 + \<g,g\> + \<h,h\> + \int_{\mathbb R_+^2} b^2(x,y)dxdy.$$

 The dual operator of $\mathcal C$ is given by
 $$\mathcal C^*f(x) = cf(0) + g(x)f(0) + \<f,h\> + \int_0^\infty q^*(x,z)f'(z)dz$$
 where $q^*(x,y) = \int_0^x\sqrt{w(y)/w(z)}b(y,z)dz$. In particular, $\mathcal C$ is symmetric if and only if $b$ is symmetric and $g=h$.
\end{thm}
\begin{proof}
 Let $H_w^c$ be the set of constant functions in $H_w$. Then $H_w$ is the orthogonal sum of $H_w^c$ and $H_w^0$ where $H_w^0$ is defined in Lemma~\ref{K:HS auf G}. Hence we have
$$H_w\otimes H_w = H_w^c\otimes H_w^c + H_w^c\otimes H_w^0+H_w^0\otimes H_w^c + H_w^0\otimes H_w^0$$
where the spaces on the right-hand side are orthogonal to each other. Thus the first claim follows from Lemma~\ref{K:HS auf G} and the value of the norm from the Pythagorean theorem. The structure of $\mathcal C^*$ follows from similar arguments.
\end{proof}

 Theorem~\ref{S:Darstellung von HS operatoren} can be used to find a representation for positive semidefinite trace-class operators. Indeed, any positive semidefinite trace class operator is the square of a symmetric Hilbert-Schmidt operator. Both types of operators play a key role in the stochastic integration theory with values in a Hilbert space, 
cf.\ Peszat and Zabczyk~\cite{peszat.zabczyk.07}.
\begin{cor}\label{K:Positive Spurklasseoperatoren}
 Let $\Q$ be a positive semidefinite trace class operator on $H_w$. Then there exist $c\in\mathbb R_+$ and a measurable function $\ell:\mathbb R_+^2\rightarrow\mathbb R$ such that $\ell$ is absolutely continuous in its first variable,\begin{enumerate}
 \item $\ell(0,\cdot)/\sqrt{w}\in L^2(\mathbb R_+)$,
 \item $(x,z)\mapsto \frac{w(z)}{w(x)}(\partial_1\ell(x,z))^2$ is symmetric and integrable,
 \item \begin{eqnarray*}
  \Q f(x) &=& \left(f(0)c + \int_0^\infty \ell(0,z)f'(z)dz\right)\left(c+\int_0^x\ell(0,z)/w(z)dz\right) \\
          &&+\ f(0)\int_0^\infty \ell(x,z)l(0,z)/w(z)dz + \int_0^\infty\int_0^\infty \ell(x,z)\partial_1\ell(z,y)dzf'(y)dy
 \end{eqnarray*}
 for any $f\in H_w$, $x\in\mathbb R_+$.
 \item \begin{eqnarray*}
   \<\Q f,f\> &=& \left(f(0)c + \int_0^\infty \ell(0,z)f'(z)dz\right)^2 \\
            &&+\ \int_0^\infty \left(f(0)\ell(0,x) + \int_0^\infty w(x)\partial_1\ell(x,z)f'(z)dy\right)^2 \frac{1}{w(x)}dx
 \end{eqnarray*}
  for any $f\in H_w$.
 \item $\tr(\Q) = c^2 + \int_0^\infty (q(0,z))^2/w(z)dz + \int_0^\infty \<q(x,\cdot),q(x,\cdot)\>/w(x)dx$
\end{enumerate}

Moreover, if $c\in\mathbb R_+$ and $\ell:\mathbb R_+^2\rightarrow\mathbb R$ is measurable, absolutely continuous in its first variable, $\ell$ satisfies (1) and (2) and 
$\Q$ is defined by (3), then $\Q$ is a positive trace class operator satisfying (4) and (5).
\end{cor}
\begin{proof}
 Let $\mathcal C$ be a symmetric root of $\Q$. Then $\mathcal C$ is a symmetric Hilbert-Schmidt operator on $H_w$. Let $b,g,h,c,q$ be as in Theorem~\ref{S:Darstellung von HS operatoren}. Since $\mathcal C$ is symmetric we can see from Theorem~\ref{S:Darstellung von HS operatoren} that $g=h$ and $b$ is symmetric. Define $\ell(x,z):=h'(z)w(z)+q(x,z)$. Then $\ell(0,\cdot)=wh'$ thus $\ell$ satisfies (1). Moreover,
 $$\frac{w(x)}{w(z)}(\partial_1\ell(x,z))^2 = b^2(x,z),\quad x,z\in\mathbb R_+$$
 thus $\ell$ satisfies (2). We also have
 \begin{eqnarray*}
  \mathcal Cf(x) &=& cf(0) + \<h,f\> + f(0)h(x) + \int_0^\infty q(x,z)f'(z) dz \\
        &=& f(0)(c+h(x)) + \int_0^\infty \ell(x,z)f'(z)dz 
 \end{eqnarray*}
 and hence
 \begin{eqnarray*}
  \Q f(x) &=& \mathcal C^2f(x) \\
   &=& \mathcal Cf(0)(c+h(x)) + \int_0^\infty \ell(x,z)(\mathcal Cf)'(z)dz \\
          &=& \left(f(0)c + \int_0^\infty \ell(0,z)f'(z)dz\right)(c+h(x)) \\
          &&+\ \int_0^\infty \ell(x,z)\left(f(0)h'(z) + \int_0^\infty \partial_1\ell(z,y)f'(y)dy\right)dz \\
          &=& \left(f(0)c + \int_0^\infty \ell(0,z)f'(z)dz\right)\left(c+\int_0^x\ell(0,z)/w(z)dz\right) \\
          &&+\ f(0)\int_0^\infty \ell(x,z)l(0,z)/w(z)dz + \int_0^\infty\int_0^\infty \ell(x,z)\partial_1\ell(z,y)dzf'(y)dy
 \end{eqnarray*}
 for any $f\in H_w$, $x\in\mathbb R_+$ where we used Fubini's theorem for the last equation. Moreover,
 \begin{eqnarray*}
   \<\Q f,f\> &=& \<\mathcal Cf,\mathcal Cf\> \\
            &=& (\mathcal Cf(0))^2 + \int_0^\infty ((\mathcal Cf)'(x))^2w(x)dx \\
            &=& \left(f(0)c + \int_0^\infty \ell(0,z)f'(z)dz\right)^2 \\
            &&+\ \int_0^\infty \left(f(0)\ell(0,x) + \int_0^\infty w(x)\partial_1\ell(x,z)f'(z)dy\right)^2 \frac{1}{w(x)}dx
 \end{eqnarray*}
 for any $f\in H_w$, $x\in\mathbb R_+$.
\end{proof}

If $\T_i$ is an integral operator with kernel $q_i$, $i=1,2,3$ and $\T_1\T_2=\T_3$, then, in many cases, one can represent $q_3$ in terms of $q_1$, $q_2$. For instance, if we let $\Psi$ in the dynamics of $f$ in \eqref{G:f dynamik} have values in the set of integral operators, then the covariance $f$ is given by
 $$ \<\< f,f\>\>(t) = \int_0^t \U_{t-s}\Psi(s)\Q\Psi(s)^*\U^*_{t-s} ds. $$
$\Psi\Q\Psi^*$ is a positive trace class operator and it can be represented by $(\Psi\C)(\Psi\C)^*$ where $\C$ is a symmetric root of $\Q$. Then $\C$ and hence $\Psi\C$ are Hilbert-Schmidt operators and hence they are 'almost' integral operators. If $\Psi$ is chosen to be an integral operator as well, then a formula for the kernel of $\Psi\C$ is obtainable.
Let us state general conditions for the connection of such operators.
\begin{thm}
\label{thm:convolution_kernels}
 Let $\T_i$ be an integral operator with kernel $q_i$, $i=1,2,3$ and $\T_3\subseteq\T_1\T_2$ where $\T_3$ is densely defined. Assume that
 \begin{enumerate}
  \item $q_2$ is absolutely continuous with respect to the first variable,
  \item $\left(z\mapsto \sup\left\{ \frac{\vert q_2(y,z)\vert}{\sqrt{w(z)}} : y\in[y_0,y_0+1] \right\}\right) \in L^2(\mathbb R_+)$ for any $y_0\in\mathbb R_+$,
  \item $\left(z\mapsto \sup\left\{ \frac{\vert \partial_1q_2(y,z)\vert}{\sqrt{w(z)}} : y\in[y_0,y_0+1] \right\}\right) \in L^2(\mathbb R_+)$ for any $y_0\in\mathbb R_+$ and
  \item $\left(z\mapsto \frac{1}{\sqrt{w(z)}}\int_0^\infty \vert q_1(x,y)\partial_1q_2(y,z)\vert dy \}\right) \in L^2(\mathbb R_+)$ for any $x\in\mathbb R_+$.
 \end{enumerate}
 Then
 $$ q_3(x,z) = \int_0^\infty q_1(x,y)\partial_1q_2(y,z)dy,\quad x,y\in\mathbb R_+. $$ 
\end{thm}
\begin{proof}
  Observe, that Conditions (1) to (3) imply, $z\mapsto \partial_1q_2(y,z)f'(z)$ and $z\mapsto q_2(y,z)f'(z)$ are integrable for any $y\in\mathbb R_+$ because $\vert f'\vert\sqrt{w}$ is square-integrable for any $f\in H_w$. Moreover, the uniform condition in $y$ allows to interchange the $z$-integration and the derivative with respect to $y$, hence
 \begin{align*}
  \int_0^\infty\partial_1q_2(y,z)f'(z) dz &= \partial_y\int_0^\infty q_2(y,z)f'(z) dz.
 \end{align*}
  Similarily, Condition (4) yields 
  $$ (y,z)\mapsto q_1(x,y)\partial_1q_2(y,z) f'(z) $$
 is integrable for any $f\in H_w$, $x\in\mathbb R_+$. Thus we have
 \begin{align*}
   \int_0^\infty q_3(x,z)f'(z) dz &= \T_3f(x) \\
                                  &= \T_1(\T_2f)(x) \\
                                  &= \int_0^\infty q_1(x,y) \partial_y\left(\int_0^\infty q_2(y,z)f'(z)dz \right) dy \\
                                  &= \int_0^\infty q_1(x,y) \int_0^\infty \partial_1q_2(y,z)f'(z)dz dy \\
                                  &= \int_0^\infty \int_0^\infty q_1(x,y) \partial_1q_2(y,z)dy f'(z) dz
 \end{align*}
 for any $f\in \dom(\T_3)$, $x\in\mathbb R_+$ where we used Fubini's theorem in the last step. Since $\dom(\T_3)$ is densely defined a monotone class argument yields
 $$ q_3(x,z) = \int_0^\infty q_1(x,y) \partial_1q_2(y,z)dy, $$
 and the Theorem follows.
\end{proof}

We end this analysis of Hilbert-Schmidt operators with looking at the covariance operator
$\Q$, and investigate the correlation implied on the L\'evy field $L(t,x)$ driving the noise
of the forward dynamics in the case $L$ is $H_w$-valued. 

Assume $\Q$ is the given as the integral operator 
$$
\Q f(x)=\int_0^{\infty}q(x,y)f'(y)\,dy
$$
for some "nice" kernel function $q$, i.e.\ the function $h_x$ defined in Lemma~\ref{l:stetigkeit der Punktauswertung} is in the domain of $\Q$. Recall that 
$$
\E[L(t,x)L(s,y)]=\E[\delta_xL(t)\delta_yL(s)]=t\wedge s\langle \Q h_x,h_y\rangle .
$$ 
We have that 
$$
h_x'(v)=\left\{\begin{array}{ll} 0, & v>x \\ w^{-1}(v), & v\leq x\,. \end{array}\right.
$$
Thus,
$$
\Q h_x(u)=\int_0^{\infty}q(u,v)\mathbf{1}(v\leq x)w^{-1}(v)\,dv=\int_0^x
q(u,v)w^{-1}(v)\,dv\,.
$$
Moreover, 
$$
(\Q h_x)'(u)=\partial_u\int_0^xq(u,v)w^{-1}(v)\,dv\,.
$$
But then for $x,y\in\mathbb R_+^2$,
\begin{align*}
\langle\Q h_x,h_y\rangle&=\Q h_x(0)h_y(0)+\int_0^{\infty}w(u)(\Q h_x)'(u)h_y'(u)\,du \\
&=\int_0^xq(0,v)w^{-1}(v)\,dv+\int_0^{\infty}w(u)\partial_u\int_0^xq(u,v)w^{-1}(v)\,dv
w^{-1}(u)\mathbf{1}(u\leq y)\,du \\
&=\int_0^xq(0,v)w^{-1}(v)\,dv+\int_0^{y}\partial_u\int_0^xq(u,v)\,w^{-1}(v)\,dvdu \\
&=\int_0^xq(y,v)w^{-1}(v)\,dv.
\end{align*}
Evidently, the covariance between $L(t,x)$ and $L(s,y)$ will depend on the choice of 
$q$ as well as the weight function $w$ of the space $H_w$. For example, choosing 
$q(y,v)=\exp(-\delta(|y-v|)$ and $w(v)=\exp(-\alpha v)$ for $\alpha>\delta>0$, then, 
supposing $y\geq x$ we find that the spatial correlation structure becomes
$$
\text{corr}(L(1,x),L(1,y))=e^{-\frac12\delta(y-x)}\sqrt{\frac{1-e^{-(\alpha-\delta)x}}{1-e^{-(\alpha-\delta)y}}}.
$$
Note that this is not stationary in distance $y-x$ between the time to maturities.

\begin{rem}
 $\Psi\Q\Psi^*$ is a positive trace-class operator and hence it has a unique positive root $\mathcal C$ which is a Hilbert-Schmidt operator. Then 
Theorem~\ref{S:Darstellung von HS operatoren} yields the existence of a kernel $k:\mathbb R_+\times\mathbb R_+\rightarrow\mathbb R$ and a function $g:\mathbb R_+\rightarrow\mathbb R$ such that
$$ \mathcal Cf(x) = g(x)f(0) + \int_0^\infty k(x,y)f'(y)dy,\quad x\in\mathbb R_+$$
for any $f\in H_w$. By a similar discussion as above, we can find the covariance (or
correlation) structure for the field $f(t,x)$ in terms of $k$.
But in the case $\Q$ and $\Psi$ are integral operators with sufficiently
regular kernels, we can associate $k$ with these kernels by appealing to 
Theorem~\ref{thm:convolution_kernels}.
\end{rem}

\subsection{Multiplication operators}\label{a:muliplikatoren}
An other useful class of operators are multiplication operators. For example,
we recall them appearing in connection with integral operators in 
Subsection~\ref{subsect:int_op}.

\begin{defn}
 Let $f:\mathbb R_+\rightarrow\mathbb R$ be a measurable function and
$$D_f:=\{g\in H_w: fg\in H_w\}.$$
The {\em multiplication operator with kernel $f$} is defined by
 $$\mathcal M_f:\mathcal D_f\rightarrow H_w,g\mapsto fg.$$
\end{defn}
We have the following: 
\begin{lem}
The multiplication operator $\mathcal M_f$ with kernel $f$ is closed.
\end{lem}
\begin{proof}
 Let $(g_n)_{n\in\mathbb N}$ be a sequence in $D_f$ which converges to some $g$ in $H_w$ such that $\mathcal M_fg_n$ converges to some $b$ in $H_w$. Let $x\in\mathbb R_+$. Then we have
\begin{eqnarray*}
   b(x) &=& \delta_x(b) \\
        &\leftarrow& \delta_x(\mathcal M_fg_n) \\
        &=& f(x)\delta_x(g_n) \\
        &\rightarrow& f(x)g(x).
\end{eqnarray*}
Thus $fg=b\in H_w$ and hence $g\in D_f$ and $\mathcal M_fg=b$.
\end{proof}
As it turns out, under some additional hypothesis on the weight function $w$, the domain $D_f$ of $\mathcal M_f$ is $H_w$ if and only if $f\in H_w$. In this case the multiplication operator is continuous. This is the content of the next Theorem: 
\begin{thm}\label{s:stetige Muliplikationsoperatoren}
Assume that $k^2:=\int_0^\infty \frac{1}{w(x)}\,dx<\infty$ and let $f:\mathbb R_+\rightarrow\mathbb R$ measurable. Then
the  following statements are equivalent.
 \begin{enumerate}
 \item $\mathcal M_f$ is a continuous linear operator on $H_w$,
 \item $\mathcal M_f$ is everwhere defined and
 \item $f\in H_w$.
 \end{enumerate}
If $\mathcal M_f$ is a continuous linear operator, then its operator norm is at most $\sqrt{5+4k^2}\Vert f\Vert_w$ and $\mathcal M_f^*g(x)=\<g,fh_x\>$ where $h_x$ is defined as in Lemma \ref{l:stetigkeit der Punktauswertung}.
\end{thm}
\begin{proof}
The implication (1)$\Rightarrow$(2) is obvious. For the implication (2)$\Rightarrow$(3),  we 
observe that $f =\mathcal M_f1 \in H_w$ where $1\in H_w$ denotes the function which is constant one. The last implication (3)$\Rightarrow$(1)
requires more care:

Let $g\in H_w$ and recall from Lemma \ref{l:punktauswertung bei undendlich} that $\Vert g\Vert_\infty\leq c\Vert g\Vert_w$ where $c:=\sqrt{1+k^2}$. We have
 \begin{eqnarray*}
   \Vert fg\Vert^2_w &=& f^2(0)g^2(0) + \int_0^\infty w(x) ((fg)'(x))^2 dx \\
                &\leq& \Vert f\Vert_w^2\Vert g\Vert_w^2 + \int_0^\infty w(x) \left((f(x)g'(x))^2+2f(x)g(x)f'(x)g'(x)+(f'(x)g(x))^2\right)dx \\
                &\leq& \Vert f\Vert_w^2\Vert g\Vert_w^2 + \Vert f\Vert_\infty^2\Vert g\Vert_w^2+\Vert f\Vert_w^2\Vert g\Vert_\infty^2+2 \Vert f\Vert_\infty\Vert g\Vert_\infty \int_0^\infty w(x) \vert f'(x)g'(x)\vert dx \\
                &\leq& (1+2c^2)\Vert f\Vert_w^2\Vert g\Vert_w^2 + 2c \Vert f\Vert_w\Vert g\Vert_w \vert\<\tilde f,\tilde g\>\vert \\
                &\leq& (1+2c^2)\Vert f\Vert_w^2\Vert g\Vert_w^2 + 2c \Vert f\Vert_w\Vert g\Vert_w \Vert \tilde f\Vert_w \Vert\tilde g\Vert_w \\
                &=& (1+4c^2)\Vert f\Vert_w^2\Vert g\Vert_w^2
 \end{eqnarray*}
 where $\tilde f(x) :=\int_0^x \vert f'(x)\vert dx$, $\tilde g(x) := \int_0^x \vert g'(x)\vert dx$, $x\in\mathbb R_+$ and we used Cauchy-Schwarz inequality for the last inequality. Thus $\mathcal M_f$ is a continuous linear operator with operator norm at most $\sqrt{1+4c^2}\Vert f\Vert_w$.

 For the representation of the dual operator simply observe that
 $$\mathcal M_f^*g(x) = \<\mathcal M_f^*g,h_x\> = \<g,\mathcal M_fh_x\> = \<g,fh_x\>.$$
This completes the proof.
\end{proof}

As an example, we consider a specific forward dynamics using a multiplication operator to define the diffusion term. Suppose that $w(x)=e^{\alpha x},x\in\mathbb R_+$ for some $\alpha >0$ for the space $H_w$. In this case the assumption $k^2=\int_0^{\infty}w^{-1}(y)\,dy=1/(2\alpha)<\infty$ is satisfied. In our example we will assume that the driving noise is a Brownian motion $W$ on the space $H_w$, and that $\Psi(t) := \psi(t,f(t))$, $\psi(t,h) := \mathcal M_{hg(t)}$ for any $h\in H_w$, $t\geq 0$ where $g:\mathbb R_+\rightarrow H_w$ is continuous. For simplicity we will assume that $\beta=0$, i.e.\ no drift in the forward dynamics. 
Theorem~\ref{s:stetige Muliplikationsoperatoren} yields that $\psi$ satisfies the requirements of Proposition~\ref{p:f ODE}. In view of Proposition~\ref{p:f ODE} there is a stochastic c\`adl\`ag process $f$ which is a mild solution to the $H_w$-valued stochastic differential equation 
$$
df(t) = \partial_xf(t)dt + \Psi(f(t))dW(t),
$$
with $f(0)=f_0\in H_w$, i.e.\
$$ 
f(t) =\U_tf_0 + \int_0^t \U_{t-s}\psi(t,f(s))\, dW(s),\quad t\in\mathbb R_+.
$$
 Let $B_T$ be as in Example \ref{b:Psi=I}, $G(t,T):=g(t,T-t)$ for $0\leq t\leq T<\infty$, then we have
 \begin{align*}
  F(t,T) &= F(0,T) + \int_0^t F(s,T)G(s,T) dB_T(s),
\end{align*} 
or
\begin{align*}
    F(t,T)     = F(0,T) \mathcal E ( A_T(t) ),
\end{align*}
where $\mathcal E( A_T(t) )$ is the stochastic exponential of the process 
\begin{align*}
  A_T(t) &:= \int_0^t G(s,T) dB_T(s).
\end{align*}
In particular,
\begin{align*} 
  S(t) &= f_0(t) + \int_0^t F(s,t)G(s,t) dB_t(s),\quad t\in\mathbb R_+.
 \end{align*}
 In this example the dynamics of the forwards are given by a mere stochastic exponential while the spot price process follows a rather complicated dynamic. $A_T$ is a Gaussian process with independent increments. For the specific choice $G(t,T) = \exp(-\delta(T-s))\sigma(T)$, $0\leq t\leq T<\infty$, $\delta>0$ we recover the forward dynamics used in Audet et al.~\cite{audet.et.al.04}. Here, $\delta$ is used to model the Samuelson effect and the time-inhomogeneous function $\sigma$ covers possible seasonality effects in the diffusion.

We end with a structural result on the space $H_{w}$.
\begin{prop}\label{k:Banachalgebra} Assume that $k^2:=\int_0^\infty \frac{1}{w(x)}\,dx<\infty$.
 Let $\Vert f\Vert:=k_1\Vert f\Vert_w$ for $f\in H_w$ where $k_1=\sqrt{5+4k^2}$. Then 
 $$ (H_w,\Vert\cdot\Vert) $$
 is a separable Banach algebra relative to the pointwise multiplication.
\end{prop}
\begin{proof}
 Let $f,g\in H_w$. Theorem~\ref{s:stetige Muliplikationsoperatoren} yields
 $$ \Vert fg\Vert = k_1\Vert\mathcal M_fg\Vert_w \leq k_1^2\Vert f\Vert_w\Vert g\Vert_w = \Vert f\Vert\Vert g\Vert.$$
\end{proof}
The Banach algebraic property of $H_w$ turns out to simplify 
considerably the proof of a technical Lemma stated earlier 
(see Lemma~\ref{l:Stetigkeit der Quadratfunktion}): We show the following Corollary on 
Lipschitz-continuity of the square function:  
\begin{cor}\label{k:Stetigkeit der Quadratfunktion}
 Assume that $\int_0^\infty\frac{1}{w(x)}dx<\infty$. Then
 $$\Vert g_1^2-g_2^2\Vert_w \leq 3\Vert h_\infty\Vert_w \Vert g_1+g_2\Vert_w \Vert g_1-g_2\Vert_w$$
 for any $g_1,g_2\in H_w$.
\end{cor}
\begin{proof}
  Define $f_+:=g_1+g_2$ and $f_-:=g_1-g_2$. Then Corollary \ref{k:Banachalgebra} yields
\begin{align*}
 \Vert g_1^2-g_2^2 \Vert_w &= \Vert f_+f_- \Vert_w 
                            \leq \sqrt{1+4k^2} \Vert f_+\Vert_w\Vert f_-\Vert_w.
\end{align*}
\end{proof}

\appendix
\section{Some technical results}

We present some technical results and considerations which are used in the main text of 
the paper. Most of these results are known, but collected here for the 
convenience of the reader.

\subsection{Riesz bases}
%
%

 Let $\lambda\in\mathbb C$, $T>0$, $g_\lambda(x):=e^{-\lambda x}$ for any $x\in[0,T)$ and
\begin{equation}
\label{r:exponentialmultiplikator}
\mathcal M_\lambda:L^2([0,T),\mathbb C) \rightarrow L^2([0,T),\mathbb C), f\mapsto fg_\lambda.
\end{equation}
 Then $\mathcal M_\lambda$ is an invertible linear operator and
 $$ \Vert f\Vert\min\{1,e^{-T\Re(\lambda)}\} \leq \Vert\mathcal M_\lambda f\Vert\leq \Vert f\Vert\max\{1,e^{-T\Re(\lambda)}\}.$$

\begin{lem}\label{l:stetige Einbettung}
  Let $\lambda,T>0$, $$\mathrm{cut}:\mathbb R\rightarrow [0,T),x\mapsto x-\max\{Tz: z\in\mathbb Z:Tz\leq x\}$$ and
  \begin{align*}
   \mathcal A:L^2([0,T),\mathbb C)&\rightarrow L^2(\mathbb R_+,\mathbb C), \\
   f&\mapsto \left(x\mapsto e^{-\lambda x}f(\mathrm{cut}(x))\right).
  \end{align*}
  Then $\mathcal A$ is a bounded linear operator and its range is closed in $L^2(\mathbb R_+,\mathbb C)$. Moreover,
  $$ \frac{e^{-2T\lambda}}{1-e^{-2T\lambda}}\Vert f\Vert^2 \leq \Vert\mathcal Af\Vert^2 \leq \frac{1}{1-e^{-2T\lambda}}\Vert f\Vert^2$$
  for any $f\in L^2([0,T),\mathbb C)$.
\end{lem}
\begin{proof}
 $\mathcal A$ is obviously linear. Let $f\in L^2([0,T),\mathbb C)$. Then
 \begin{align*}
   \Vert Af\Vert^2  & = \sum_{n=0}^\infty \int_{nT}^{nT+T} e^{-2\lambda x} \vert f\vert^2(\mathrm{cut}(x)) dx \\
                    &= \sum_{n=0}^\infty \int_0^{T} e^{-2\lambda x-2\lambda T n} \vert f\vert^2(x) dx \\
                    &= \sum_{n=0}^\infty e^{-2\lambda T n} \Vert\mathcal M_\lambda f\Vert^2
 \end{align*}
 where $\mathcal M_\lambda$ is defined in \eqref{r:exponentialmultiplikator}. The norm
estimates on $\mathcal M_{\lambda}$ yields
 \begin{align*}
   \Vert\mathcal Af\Vert^2 &\leq \sum_{n=0}^\infty e^{-2\lambda Tn}\Vert f\Vert^2 
                   = \frac{1}{1-e^{-2T\lambda}}\Vert f\Vert^2
 \end{align*}
 and hence $\mathcal A$ is bounded by $\sqrt{\frac{1}{1-e^{-2T\lambda}}}$. Moreover, the computation at the beginning of this section also implies
 \begin{align*}
   \Vert\mathcal Af\Vert^2 &\geq \sum_{n=0}^\infty e^{-2\lambda T(n+1)}\Vert f\Vert^2 
                   =   \frac{e^{-2\lambda T}}{1-e^{-2\lambda T}}\Vert f\Vert^2.
 \end{align*}
  Hence $\mathcal A$ is an isomorphism on its range. Consequently, the range of 
$\mathcal A$ is closed.
\end{proof}

\begin{defn}\label{d:Einschraenkungsoperator}
 Let $T>0$. The {\em restriction operator} is the continuous linear projection given by
  $$\mathcal R_T :L^2(\mathbb R_+,\mathbb C)\rightarrow L^2([0,T),\mathbb C), f \mapsto f\vert_{[0,T)}.$$
\end{defn}

\begin{lem}\label{l:Riesz basis}
 Let $\lambda,x_0>0$, define $g_n:\mathbb R_+\rightarrow\mathbb C,x\mapsto \frac{1}{\sqrt{x_0}}e^{(\frac{2\pi i n}{x_0}-\lambda) x}$ and let $V$ be the closed subspace of $L^2(\mathbb R_+,\mathbb C)$ generated by $(g_n)_{n\in\mathbb Z}$. Then the following statements hold.
 \begin{enumerate}
  \item $(g_n)_{n\in\mathbb Z}$ is a Riesz basis of $V$ and
  \item The linear operator $\mathcal B:V\rightarrow L^2([0,x_0],\mathbb C), f\mapsto f\vert_{[0,x_0]}$ is invertible and continuous.
 \end{enumerate}
\end{lem}
\begin{proof}
 \begin{enumerate}
  \item For any $n\in\mathbb Z$ define
 $$ e_n:[0,x_0]\rightarrow\mathbb C,x\mapsto \frac{1}{\sqrt{x_0}}e^{\frac{2\pi in}{x_0} x}.$$
 Then $(e_n)_{n\in\mathbb Z}$ is an orthonormal basis for $L^2([0,x_0],\mathbb C)$. Let $\mathcal A$ be defined as in Lemma \ref{l:stetige Einbettung}. Then 
$g_n = \mathcal Ae_n$ for any $n\in\mathbb Z$. Hence the closed subspace generated by $(g_n)_{n\in\mathbb Z}$ is the range of $\mathcal A$. The properties of 
$\mathcal A$ which are stated in Lemma \ref{l:stetige Einbettung} and Young~\cite[Theorem 1.9]{young.80} imply that $(g_n)_{n\in\mathbb Z}$ is a Riesz basis of $V$.
  \item Define $\mathcal B^-:=\mathcal A\circ \mathcal M_{-\lambda}$ where 
$\mathcal M_{-\lambda}$ is defined in \eqref{r:exponentialmultiplikator}. Then 
$\mathcal B$ is continuous, linear and invertible as an operator from $L^2([0,x_0],\mathbb C)$ to $V$. Its inverse $\mathcal B$ is also an invertible continuous linear operator and for $n\in\mathbb N$ we have $\mathcal B g_n = \mathcal M_\lambda e_n = g_n\vert_{[0,x_0]}$.
 \end{enumerate}
\end{proof}

\subsection{Properties of the stochastic integral}
\begin{prop}\label{P:Eig. quad int}
 Let $U$ be a Hilbert space, $T\in\mathbb R_+$, $\mathcal M^2_T$ be the set of $U$-valued square integrable martingales with a.s.\ c\`adl\`ag paths defined on the time intervall $[0,T]$ and $\mathcal M^2_{T,c}$ be the subset of $\mathcal M^2_T$ which contains all a.s.\ continuous square integrable martingales. Let 
$$\Vert\cdot\Vert_T:\mathcal M^2_T\rightarrow\mathbb R,M\mapsto \Vert M\Vert_T:=\sqrt{\E(M^2(T))}.$$
Then $\mathcal M^2_{T,c}$ is a closed subset of the pre-Hilbert space $(\mathcal M^2_T,\Vert\cdot\Vert_T)$. Moreover, $\Vert\cdot\Vert_T$ is equivalent to the norm
$$\Vert\cdot\Vert_T^*:\mathcal M^2_T\rightarrow\mathbb R,M\mapsto \Vert M\Vert_T:=\sqrt{\E(\sup\{M^2(t)):t\in[0,T]\})}.$$
Let $M,N\in\mathcal M^2_T$. Then $\Vert M-N\Vert_T=0$ if and only if there is a null set $N\subseteq\Omega$ such that
$$N(\omega,\cdot)=M(\omega,\cdot)$$
for all $\omega\in\Omega\backslash N$.
\end{prop}
\begin{proof}
  See Pr\'ev{\^ o}t and R\"ockner~\cite[Chapter 2]{Prevot.Roeckner.07}.
\end{proof}

As a simple consequence of Proposition \ref{P:Eig. quad int} one has the well-known
\begin{cor}\label{C:Pfadeigenschaft}
 Let $M$ be a square integrable martingale with a.s.\ c\`adl\`ag paths and values in some Hilbert-space $U$, $H$ be a Hilbert space and $\Psi\in\mathcal L^2_M(H)$. Then
 $$X(t) := \int_0^t\Psi(s)dM(s)$$
defines a square integrable martingale with a.s.\ c\`adl\`ag paths. If $M$ has a.s.\ continuous paths, then $X$ has a.s.\ continuous paths.
\end{cor}

\begin{prop}\label{P:Assoziativitaet}
 Let $M$ be a square integrable martingale with a.s.\ c\`adl\`ag paths and values in some Hilbert-space $U$, $F,H$ be Hilbert spaces, $\Psi\in\mathcal L^2_M(F)$, $I(t):=\int_0^t\Psi(s)dM(s)$, $t\in\mathbb R_+$ and $\Gamma\in\mathcal L^2_{I}(H)$. Then 
$$\Phi(t):=\Gamma(t)\Psi(t),\quad t\in\mathbb R_+$$
is an element of $L^2_M(H)$ and
$$ \int_0^t \Phi(s) dM(s) = \int_0^t \Gamma(s) dI(s).$$
\end{prop}
\begin{proof}
 We outline the proof taken from Pr\'ev{\^ o}t and R\"ockner~\cite[Chapter 2]{Prevot.Roeckner.07}:
First assume that $\Gamma$ and $\Psi$ are simple. Then $\Phi$ is simple and the equality follows from an elementary computation. General $\Gamma$ and $\Psi$ can be approximated by simple integrands such that $L^2(\Omega,H)$-convergence holds.
\end{proof}



\begin{thebibliography}{99}

\bibitem{andresen.et.al.10} Andresen, A., Koekebakker, S. and Westgaard, S. (2010). Modeling electricity forward prices using the multivariate normal inverse Gaussian distribution. {\it J. Energy Markets.} {\bf 3}(3).

\bibitem{Arveson.02} Arveson, W. (2002). {\it A Short Course on Spectral Theory}, Springer,
New York.

\bibitem{audet.et.al.04} Audet, A., Heiskanen, P., Keppo J. and Vehvil\"ainen A. (2004). Modeling Electricity Forward Curve Dynamics in the Nordic Market. In {\it Modelling Prices in Competitive
Electricity Markets.},  Wiley \& Sons, Inc. Pages 251-265

\bibitem{BN} Barndorff-Nielsen, O.E. (1998). Processes of normal inverse Gaussian
type. {\it Finance \& Stoch.}, {\bf 2}(1), pp.~41--68.

\bibitem{BNBV-spot} Barndorff-Nielsen, O. E., Benth, F. E., and Veraart, A. (2013).
Modelling energy spot prices by volatility modulated L\'evy-driven Volterra
processes. {\it Bernoulli}, {\bf 19}(3), pp.~803--845.

\bibitem{BNSch} Barndorff-Nielsen, O. E., and Schmiegel, J. (2004). L\'evy-based tempo-spatial modelling; with applications to turbulence.  {\it Uspekhi Mat. NAUK}, {\bf 59}, pp.~65--91.

\bibitem{BB} Barth, A., and Benth, F. E. (2010). The forward dynamics in energy markets --
infinite dimensional modeling and simulation. To appear in {\it Stochastics}.

\bibitem{B-sv} Benth, F. E. (2011). The stochastic volatility model of Barndorff-Nielsen and
Shephard in commodity markets. {\it Math. Finance}, {\bf 21}, pp.~595--625.

\bibitem{BSBK-book} Benth, F.E., \v{S}altyt\.e Benth, and Koekebakker, S. (2008). 
{\it Stochastic Modelling of Electricity and Related Markets}, World Scientific, Singapore.

\bibitem{BSB-book} Benth, F.E., and \v{S}altyt\.e Benth. (2013). 
{\it Modeling and Prcing in Financial Markets for Weather Derivatives}, World Scientific, Singapore.

\bibitem{benth.et.al.05} Benth, F.E., Kallsen, J., and Meyer-Brandis, T. (2007).
A non-Gaussian Ornstein-Uhlenbeck process for electricity spot price modeling and derivatives 
pricing. {\it Appl. Math. Finance}, {\bf 14}, 153--169.


\bibitem{BK} Benth, F. E., and Koekebakker, S. (2008). Stochastic modeling of financial electricity contracts. {\it Energy Econ.}, {\bf 30}(3), pp.~1116--1157.


\bibitem{benth.kruehner.13} Benth, F.E. and Kr\"uhner, P. (2013). Subordination of 
Hilbert-space valued L\'evy processes. Preprint downloadable at http://arxiv.org/abs/1211.6266

\bibitem{BL} Benth, F. E., and Lempa, J. (2014). Optimal portfolios in commodity 
futures markets. To appear in {\it Finance Stoch.}. DOI: 10.1007/s00780-013-0224-5 

\bibitem{BG} Bj\"ork, T., and Gombani, A. (1999). Minimal realizations of interest rate
models. {\it Finance Stoch.}, {\bf 3}, pp.~413--432.

\bibitem{BCKS}  B\"orger, R., Cartea, A., Kiesel, R., and Schindlmayr, G. (2009). Cross-commodity analysis and applications to risk management. {\it J. Futures Markets}, {\bf 29}, 
pp.~197--217. 

\bibitem{Brock} Brockwell, P. J. (2001). Continuous-Time ARMA process. In C. R. Rao and D. N. Shanbhag (eds.), {\it Handbook of Statistics: Stochastic Processes, Theory and Methods}, Elsevier, Amsterdam, pp.~249--276.


\bibitem{buehler.06} B\"uhler, H. (2006). Consistent variance curve models. {\it Finance \& Stoch.},
{\bf 10}, 178--203

\bibitem{carmona.tehranchi.06} Carmona, R. and Tehranchi, M. (2006).
{\it Interest Rate Models: An Infinite Dimensional Stochastic Analysis Perspective}, Springer,
Berlin Heidelberg New York.

\bibitem{carmona.nadtochiy.12} Carmona, R. and Nadtochiy, S. (2012). Tangent L{\'e}vy market models.
{\it Finance \& Stoch.}, {\bf 16}(1), 63--104.

\bibitem{CS} Clewlow, L., and Strickland, C. (2000). {\it Energy Derivatives: Pricing and Risk
Management}, Lacima Publications.

\bibitem{delbaen.schachermayer.98} Delbaen, F. and Schachermayer, W. (1998). The fundamental theorem of asset pricing for unbounded stochastic processes. {\it Mathem. Annalen}, {\bf 312}, 215--250

\bibitem{ET} Ekeland, I., and Taflin, E. (2005). A theory of bond portfolios. 
{\it Ann. Appl. Probab.}, {\bf 15}, pp.~1260--1305. 

\bibitem{FT1} Filipovi\'c, D., and Teichmann, J. (2003). Existence of invariant 
manifolds for stochastic equations in infinite dimensions. {\it J. Funct. Anal.}, {\bf 197},
pp.~398--432.

\bibitem{FT2} Filipovi\'c, D., and Teichmann, J. (2003). Regularity of 
finite-dimensional realizations for evolution equations. {\it J. Funct. Anal.}, {\bf 197},
pp.~433--446.

\bibitem{filipovic.01} Filipovi\'c, D. (2001). {\it Consistency Problems for Heath-Jarrow-Morton 
Interest Rate Models}, Lecture Notes in Mathematics, Vol. 1760, Springer, Berlin.

\bibitem{filipovic.et.al.09} Filipovi\'c, D., Teichmann, D. and Tappe, S. (2009). Term structure 
models driven by Wiener process and Poisson measures: existence and positivity. Preprint
downloadable from: http://arxiv.org/abs/0905.1413.

\bibitem{benth.et.al.10} Frestad, D., Benth, F. E., and Koekebakker, S. (2010). Modeling term structure dynamics in the Nordic electricity swap market. {\it The Energy J.} {\bf 21}(2), 53--86.


\bibitem{GKM} Garcia, I., Kl\"uppelberg, C. and M\"uller, G. (2010). Estimation of stable CARMA models with an application to electricity spot prices. {\it Statist. Mod.}, {\bf 11}(5), 
pp.~447--470.


\bibitem{grafakos.08} Grafakos, L. (2008). {\it Classical Fourier Analysis}, Sec. Ed., Springer,
New York.

\bibitem{HLC} H\"ardle, W., and Lopez Cabrera, B. (2012). The implied market price of 
weather risk. {\it Appl. Math. Finance}, {\bf 19}(1), pp.~59--95.


\bibitem{heath.al.92} Heath, D., Jarrow, R. and Morton, A. (1992). Bond pricing and the term 
structure of interest rates: a new methodology for contingent claims valuation. {\it Econometrica},
{\bf 60}, 77--105.

\bibitem{H} Hull, J. C. (2000). {\it Options, Futures \& Other Derivatives}, 4th Ed, Prentice Hall.


\bibitem{jacod.79} Jacod, J. (1979). {\it Calcul Stochastique et Probl\`emes de Martingales}.
Lecture Notes in Mathematics, Vol 714, Springer, Berlin.

\bibitem{js.87} Jacod, J. and Shiryaev, A. (2003). {\it Limit Theorems for Stochastic Processes}. Second edition, Springer, Berlin.

\bibitem{kallsen.kruehner.13} Kallsen, J. and Kr\"uhner, P. (2014). On a {H}eath-{J}arrow-{M}orton approach for stock options. To appear in {\it Finance and Stochastics}.

\bibitem{KO} Koekebakker, S., and Ollmar, F. (2005). Forward curve dynamics in the Nordic
electricity market. {\it Manag. Finance}, {\bf 31}(6), pp.~74--95.

\bibitem{lucia.schwart.02} Lucia, J. and Schwartz, E. (2002). Electricity prices and power 
derivatives: evidence from the Nordic Power Exchange. {\it Rev. Derivatives Res.}, {\bf 5}(1), 5--50.

\bibitem{palmer.94} Palmer, T. (1994). {\it Banach Algebras and the General Theory of *-Algebras},
Vol. I, Cambridge University Press, Cambridge.

\bibitem{PP} Paschke, R., and Prokopczuk, M. (2010). Commodity derivatives valuation with autoregressive and moving average components in the price dynamics. {\it  J. Banking 
Finance}, {\bf 34}(11), pp.~2741--2752. 

\bibitem{peszat.zabczyk.07} Peszat, S. and Zabczyk, J. (2007). {\it Stochastic Partial Differential Equations with L\'evy Noise}, Cambridge University Press,
Cambridge.

\bibitem{Prevot.Roeckner.07} Pr\'ev{\^ o}t, C. and R\"ockner, M. (2007). {\it A Concise Course on 
Stochastic Partial Differential Equations}, Springer, Berlin.

\bibitem{R} Rydberg, T.H. (1997). The normal inverse Gaussian L\'evy process: Simulation
and approximation. {\it Commun. Statist. -- Stochastic Models}, {\bf 13}(4), pp.~887--910.

\bibitem{sato.99} Sato, K. (1999). {\it L\'evy Processes and Infinitely Divisible Distributions},
Cambridge University Press, Cambridge.

\bibitem{Tappe.10} Tappe, S. (2010). An alternative approach on the existence of affine realizations for HJM term structure models. {\it Proc. R. Soc., Ser. A}, {\bf 466}, 
pp.~3033--3060. 

\bibitem{Tappe.12} Tappe, S. (2012). Some refinements of existence results for
SPDEs driven by Wiener proceses and Poisson random measure. {\it Intern. J. Stoch. Analysis}, 
{\bf 24}.

\bibitem{young.80} Young, R. (1980). {\it An Introduction to Nonharmonic Fourier Series},
Academic Press, Inc., New York.

\end{thebibliography}

\end{document}